%% file: main.tex
\documentclass[a4paper,UKenglish,cleveref, autoref,thm-restate]{lipics-v2021}
\title{The 2-Dimensional Constraint Loop Problem is Decidable} %TODO Please add

\author{Quentin Guilmant}{Max Planck Institute for Software Systems, Saarbrücken, Germany \and \url{https://quentin.guilmant.fr} }{quentin.guilmant@mpi-sws.org}{https://orcid.org/0009-0004-7097-0595}{}%TODO mandatory, please use full name; only 1 author per \author macro; first two parameters are mandatory, other parameters can be empty. Please provide at least the name of the affiliation and the country. The full address is optional. Use additional curly braces to indicate the correct name splitting when the last name consists of multiple name parts.

\author{Engel Lefaucheux}{ Loria, Nancy, France \and \url{https://elefauch.github.io}}{engel.lefaucheux@inria.fr}{0000-0003-0875-300X}{}

\author{Jo\"el Ouaknine}{Max Planck Institute for Software Systems, Germany}{joel@mpi-sws.org}{https://orcid.org/0000-0003-0031-9356}{}

\author{James Worrell}{Department of Computer Science, University of Oxford, United Kingdom}{jbw@cs.ox.ac.uk}{https://orcid.org/0000-0001-8151-2443}{}

\authorrunning{Q. Guilmant, E. Lefaucheux, J. Ouaknine and J. Worell} %TODO mandatory. First: Use abbreviated first/middle names. Second (only in severe cases): Use first author plus 'et al.'

\Copyright{Quentin Guilmant, Engel Lefaucheux, Jo\"el Ouaknine and James Worell} %TODO mandatory, please use full first names. LIPIcs license is "CC-BY";  http://creativecommons.org/licenses/by/3.0/

 %TODO mandatory: Please choose ACM 2012 classifications from https://dl.acm.org/ccs/ccs_flat.cfm 

 %TODO mandatory; please add comma-separated list of keywords

\category{} %optional, e.g. invited paper

\relatedversion{} %optional, e.g. full version hosted on arXiv, HAL, or other respository/website
\nolinenumbers %uncomment to disable line numbering

\EventEditors{John Q. Open and Joan R. Access}
\EventNoEds{2}
\EventLongTitle{42nd Conference on Very Important Topics (CVIT 2016)}
\EventShortTitle{CVIT 2016}
\EventAcronym{CVIT}
\EventYear{2016}
\EventDate{December 24--27, 2016}
\EventLocation{Little Whinging, United Kingdom}
\EventLogo{}
\SeriesVolume{42}
\ArticleNo{23}
\usepackage[anglais, submissionMode, parindent]{perso}
\usepackage{joeletalmaccros}
\DeclareMathOperator{\convh}{ConvHull} 
\DeclareMathOperator{\affh}{AffHull} 
\DeclareMathOperator{\rec}{rec} 
\DeclareMathOperator{\ri}{ri}
\DeclareMathOperator{\Cone}{cone} 

\makeatletter
\let\orgdescriptionlabel\descriptionlabel
\renewcommand*{\descriptionlabel}[1]{%
	\let\orglabel\label
	\let\label\@gobble
	\phantomsection
	\edef\@currentlabel{#1}%
	\let\label\orglabel
	\orgdescriptionlabel{#1}%
}
\makeatother

\begin{document}

\maketitle

\begin{abstract}
  A linear constraint loop is specified by a system of linear
  inequalities that define the relation between the values of the
  program variables before and after a single execution of the loop
  body.  In this paper we consider the problem of determining whether
  such a loop terminates, i.e., whether all maximal executions are
  finite, regardless of how the loop is initialised and how the
  non-determinism in the loop body is resolved.  We focus on the
  variant of the termination problem in which the loop variables range
  over $\mathbb{R}$.  Our main result is that the termination problem
  is decidable over the reals in dimension~2.  A more abstract
  formulation of our main result is that it is decidable whether a
  binary relation on $\mathbb{R}^2$ that is given as a conjunction of
  linear constraints is well-founded.
  \end{abstract}

\section{Introduction}
\label{sec:intro}
\input{intro}

\section{Preliminaries}
\subsection{Key notations}
\input{notations.tex}

\subsection{Convex Sets}\label{sec:prelim}\label{sec:convex}
\input{prelim}

\subsection{Minkowski-Weyl Convex Sets}
\label{sec:mwconvex}
\input{mwConvex}

\subsection{Accumulation Expansions}

\input{accExp}

\section{Deciding the Constraint Loop Problem}

\input{overview}

\subsection{Deciding the Existence of a Bounded Sequence}
\label{sec:bounded}

\input{bounded}

\subsection{A Sufficient Condition for the Existence of a Sequence}\label{sec:suff}
\input{sufficient}

\subsection{Necessary Condition for the Existence of a 1-Dimensional Sequence}
\label{sec:1D}

\input{dim1}

\subsection{Necessary Condition for the Existence of a 2-Dimensional Sequence}
\label{sec:2D}

\input{dim2}

%%
%% Bibliography
%%

%% Please use bibtex, 
\bibliography{main}

\newpage
\appendix

%\section{Proofs on Accumulation Expansions}\label{sec:proofaccExp}
%\input{appendices/accdirannex}

\section{Detailed proof for the Proposition \ref{prop:dim2nec}}
\label{sec:detailedProof}
\input{proofWitness.tex}

\end{document}

%% file: intro.tex
The problem of deciding loop termination is of fundamental importance
in software verification.  
Deciding termination is already challenging
for very simple classes of programs.  One such class consists of
\emph{linear constraint loops}.  These are single-path loops in which
both the loop guard and the loop update are given by conjunctions of
linear inequalities over the program variables.  Such a loop can be
written as follows, where $B$, $A$ are matrices of rational numbers,
$\bm{a}$, $\bm{b}$ are vectors of rational numbers, and
$\bm{x},\bm{x}'$ represent the respective values of the program
variables before and after the loop update:
\begin{equation*}
	\P \, : \, \mathsf{while} \; (B\, \bm{x} \geq \bm{b}) \; \mathsf{do} \;  A\, \begin{psmallmatrix*}[l] \bm{x}\\ \bm{x'}\end{psmallmatrix*} \geq \bm{a},
      \end{equation*}
      
      Such loops are inherently non-deterministic, since the effect of
      the loop body is described by a collection of constraints.  Note
      in passing that the loop guard can folded into the constraints
      that describe the loop body and so, without loss of generality,
      the guard can be assumed to be trivial.  Linear constraint loops
      naturally arise as abstractions of other programs.  For example,
      linear constraints can be used to model size changes in program
      variables, data structures, or terms in a logic
      program (see, e.g.~\cite{LindenstraussS97}).

A linear constraint loop is said to \emph{terminate} if there is no initial value
of the loop variables from which the loop has an infinite execution.
The Termination Problem asks to decide whether a given loop
terminates.  As such, the Termination Problem depends on the numerical
domain that the program variables range over: typically one considers
either $\Zbb$, $\Qbb$, or $\Rbb$.

One approach to proving termination of linear constraint loops
involves synthesizing linear ranking functions~\cite{BG14}.  However,
it is well-known that there are loops that terminate that admit no
linear ranking function.  In the special case of deterministic linear
constraint loops (i.e., where the loop body applies an affine function
to the program variables) decidability of termination over $\Rbb$ was
shown by Tiwari~\cite{Tiwari04}, decidability of termination over
$\Qbb$ was shown by Braverman~\cite{Braverman06}, and decidability of
termination over $\Zbb$ was established
in~\cite{HOW19}.\footnote{These works in fact consider loop guards
  that feature a mix of strict and non-strict inequalities, whereas in
  the present paper we consider only non-strict inequalities.}  All
three papers build on an analysis of the spectrum of the matrix that
determines the update function in the loop body.  To the best of our
knowledge, decidability of termination of linear constraint loops over
$\Rbb$, $\Qbb$, and $\Zbb$ remains open.  It is known however that
termination for multi-path constraint loops is undecidable (i.e.,
where disjunctions are allowed in the linear constraints that define
the update map).  It is moreover known that termination of single-path
constraint loops is undecidable if irrational constants are allowed in
the constraints~\cite{BGM12}.  One of the few known positive results
is the restricted case that all the constraints are octagonal, in
which case termination is decidable over integers~\cite{bozga}.
(Recall that a constraint is said to be \textit{octagonal} if it is a
conjunction of propositions of the form \(\pm x_i\pm x_j\leq a\),
for variables $x_i,x_j$ and constant $a\in\Zbb$.)

In this paper we study the termination of linear constraint loops over
the reals in dimension at most 2.  We give a sufficient and necessary
condition that such a loop be non-terminating in the form of a
\emph{witness of non-termination}.  This is given in
Definition~\ref{def:wit}.  Here one should think of $K$ as the
transition relation of a linear constraint loop, while $\rec(K)$ is
the \emph{recession cone} of $K$, i.e., the set of vectors $v$ such
that $w+\lambda v\in K$ for all $w\in K$ and $\lambda \geq 0$.  The
witness of non-termination is essentially a finite representation of
an infinite execution of the loop in the spirit of the geometric
non-termination arguments of~\cite{LeikeH18} and the recurrent sets
of~\cite{Ben-AmramDG19}.
\begin{definition}\label{def:wit}
  Let $E$ be a Euclidean space.  Let $K\subseteq E^2$ be a convex
  set.  A witness $\Wcal(K)$ consists of a linear map
	$M:E\to E$, a closed cone $C\subseteq E$, and $v,w\in E$, such that
	\begin{description}
		\item[($\exists u1$)\label{it:MCC}] $MC\subseteq C$
		\item[($\exists u2$)\label{it:xMx}] $\forall x\in C\qquad (x,Mx)\in \rec (K)$
		\item[($\exists u3$)\label{it:vwz}] $(v,w)\in K$
		\item[($\exists u4$)\label{it:vwzC}] $w-v\in C$.
	\end{description}
      \end{definition}
      
      If $E$ has dimension at most $2$ and $K$ is a polyhedron, then
      the existence of such a witness can be expressed 
      in the theory of real closed fields.  (The
      restriction to dimension 2 entails that every cone is generated
      by a most 3 vectors, whereas there is no such upper bound in
      dimension 3.)  Thus we obtain a polynomial-time reduction of the
      Termination Problem for constraint loops to the decision problem
      for the theory of real closed fields with a bounded number of quantifier, which is
      decidable in polynomial space.
      
      The following is our main result, which characterises
      non-termination in terms of the above notion of witness.  We
      refer to Section~\ref{sec:mwconvex} for the notion of MW-convex
      set, suffice to say here that this class includes all polyhedra
      and that main property of MW-convex sets used in the proof is
      that for every linear projection $\pi$ and MW-convex set $K$ we
      have $\pi(\rec(K)) = \rec(\pi(K))$.  Further background about
      convex sets is contained in Section~\ref{sec:convex}.

\begin{restatable}{theorem}{dimun}\label{thm:dim12}
	Let $E$ be a Euclidean space of dimension at most $2$.
	Let $K\subseteq E^2$ be MW-convex. There is a sequence
        $\suiten\in E^\Nbb$ such that 
        $(u_n,u_{n+1})\in K$ 	for all $n\in\Nbb$
	if and only if there exists a witness $\Wcal(K)$.
\end{restatable}

%% file: notations.tex
In this very short section, we introduce the notation we will use for the entire paper.

\subparagraph{Some sets}
We put an $\ast$ on sets to remove $0$ from this sets. Namely, $\Rbb^*=\Rbb\setminus\{0\}$, $\Nbb^*=\Nbb\setminus\{0\}$ and so on. $\Rbbplus$ stand for all non-negative real numbers and $\Rbbplusstar$ for all the positive real numbers. Also, for $n,m\in\Nbb$ such that $n\leq m$, we let $\intn nm$ be the set of integer in between $n$ and $m$ inclusively, namely $\intn nm=\{n,n+1,\dots,m\}$.

\subparagraph{Landau Notations}
We use the Landau notations. Let $d\in\Nbb^*$. Let $\norm{\cdot}{}$ be any norm over $\Rbb^d$ (they are equivalent anyway).  Let $u:\Nbb\to\Rbb^d$, $w:\Nbb\to\Rbb^d$ and $v:\Nbb\to\Rbb$ be sequences. We then have the following notations:
\begin{itemize}
	\item $u_n=\petito{n\to \pinf}{v_n}$ when for all $\epsilon\in\Rbbplusstar$ there is some $N\in\Nbb$ such that for all $n\geq N$, we have $\norm{u_n}{}\leq \epsilon |v_n|$.
	
	\item $u_n=\grando{n\to \pinf}{v_n}$ when there is  some $M \in\Rbbplusstar$ and some some $N\in\Nbb$ such that for all $n\geq N$, we have $\norm{u_n}{}\leq M |v_n|$.
	
	\item $u_n=\asOmega{n\to \pinf}{v_n}$ when there is  some $M \in\Rbbplusstar$ and some some $N\in\Nbb$ such that for all $n\geq N$, we have $\norm{u_n}{}\geq M |v_n|$.
	
	\item $u_n\underset{n\to\pinf}{\sim} w_n$ if $u_n-w_n=\petito{n\to\pinf}{\norm{w_n}{}}$. 
	
	\item $u_n = w_n+\petito{n\to\pinf}{v_n}$ if $u_n-w_n=\petito{n\to\pinf}{v_n}$.
	
	\item $u_n = w_n+\grando{n\to\pinf}{v_n}$ if $u_n-w_n=\grando{n\to\pinf}{v_n}$.
\end{itemize}

We keep the same notations if the sequences are undefined at a finite number of points in~$\Nbb$.

%% file: prelim.tex
%A convex set is one which contains the line segment joining any two of
%its elements.  Convexity of a set provides many useful properties that
%we use extensively within this work.  Therefore we start by providing
%a collection of convexity properties that will be needed later.
%Although the results are folklore we give certain proofs in the
%appendix for the convenience of the reader.
 
Throughout this section $E$ is an arbitrary Euclidean space.
%Considering $S\subseteq E$, we define the smallest affine subspace,
%convex set, and vector subspace of $E$ that contain $S$.

These results are already known but for the sake of completeness, some proof are
written here anyway.

\begin{definition}%[Affine/convex hull and Vector space]
\label{def:affineHull}\label{def:convHull}\label{def:Vect}
	Let $S\subseteq E$.
	The \textbf{affine hull} of $S$, denoted $\affh (S)$, the
        \textbf{convex hull} of $S$, denoted $\convh (S)$, and the
        \textbf{vector space spanned by} $S$, denoted $\Vect (S)$, are
        defined by
\begin{eqnarray*}
        \affh (S) &=& \left\{ \sum_{i=1}^k \alpha_i x_i  \mid \alpha_i 
          \in \Rbb, x_i \in S, \sum_{i=1}^k \alpha_i = 1\right\}\\
	\convh (S) &=&\left\{ \sum_{i=1}^k \alpha_i x_i  \mid \alpha_i 
          \in [0;1], x_i \in S, \sum_{i=1}^k \alpha_i = 1\right\}\\
	\Vect (S) &=&\left\{ \sum_{i=1}^k \alpha_i x_i  \mid \alpha_i 
          \in  \mathbb{R}, x_i \in S \right\}
\end{eqnarray*}
  \end{definition}

%  Note that the difference between the affine hull, the convex hull,
 % and the vector space is limited to the range of the variables
 % $\alpha_i$.

\begin{definition}%[Relative interior]
\label{def:ri}
	Let $K\subseteq E$ be a convex set. The \textbf{relative interior} of $K$, denoted $\ri (K)$, is defined by:
	$$
	\ri (K) = \enstq{x\in K}{\exists U\in \Ocal(E), \pa{x\in U}\wedge \pa{U\cap\affh (K)\subseteq K}}
	$$
	where  $\Ocal(E)$ stands for the set of open subsets of $E$.
\end{definition}

In other words, the relative interior of a convex set $C$ is its interior
with respect to the induced topology on the affine subspace spanned by $C$.

We have the following properties for the relative interior:

\begin{proposition}%[Folklore]
\label{prop:riNonEmpty} 
	Let $K\subseteq E$ be a non-empty convex set.
	Denoting as usual by $\bar K$ the smallest closed subset of $E$ containing $K$, we have:	
	\begin{romanenumerate}
		\item $\ri (K)$ is a non-empty convex set
		\item $\ri (K)\subseteq K \subseteq \bar K$
		\item $\affh (\ri (K)) = \affh (K)$
		\item $\ri (K) = \ri (\bar K)$
		\item $\bar{\ri (K)} = \bar K$
	\end{romanenumerate}
\end{proposition}

\begin{proposition}%[Folklore]
\label{prop:riBorder}
	Let  $K$ be a non-empty convex set and $x,y$ such that $x\in\ri (K)$ and $y\in K\setminus \ri (K)$. Then
	for all $\lambda\in\intof01$ we have $\lambda x + (1-\lambda)y\in \ri (K)$.
\end{proposition}

\begin{definition}
	Let $K\subseteq E$ be a non-empty convex set. The \textbf{recession cone} of $K$, denoted $\rec (K)$, is the set
	$
	\rec (K) = \enstq{z\in E}{K+\Rbb_+z\subseteq K}
	$.
\end{definition}

Note that we always have $0\in\rec (K)$. Also, the recession cone is indeed a cone, as it is
stable under positive scalar multiplication by definition.

\begin{restatable}%[Folklore]
{lemma}{projcRecc}\label{lem:projcRecc}
	Let $K\subseteq E$ be a convex set. Let $\pi: E\to E$ be a
        linear projection.  Then
	$
	\pi(\rec (K)) \subseteq \rec (\pi(K))
	$. 
\end{restatable}

\begin{proof}
	Let $x\in\pi\pa{\rec (K)}$. There is $y\in\Kernel\pi$ such that $x+y\in\rec (K)$. Let $a\in\pi(K)$ and $b\in\Kernel\pi$ such that $a+b\in K$. Then, 
	$$\forall\lambda\in\Rbb_+ \qquad (a+b)+\lambda(x+y)\in K$$
	\lc{Hence,}{$\forall\lambda\in\Rbb_+\qquad a+\lambda x\in \pi(K)$}
	\lc{and}{$x\in \rec(\pi(K))$}
\end{proof}

If $K$ is closed, we even have an alternative characterization of the recession cone which requires a seemingly weaker property but that turns out to be equivalent.

\begin{restatable}%[Folklore]
{proposition}{recConeClosedConvex}
	\label{prop:recConeClosedConvex}
	Let $K\subseteq E$ be a non-empty closed convex set. Then 
	$$
	\rec (K) = \enstq{z\in E}{\exists x\in K\quad x+\Rbb_+z\subseteq K}
	$$ 
\end{restatable}

\begin{proof}
	We proceed by double inclusion.
	\begin{itemize}
		\item[\CNsubset] This direction is easy : if for all $x\in K$, $x+\Rbb_+z\subseteq K$, since $K\neq\emptyset$, there is at least one $x$ such that $x+\Rbb_+z\subseteq K$.
		
		\item[\CSsubset] Let $z\in\Rbb^d$ such that there is some $x\in K$ such that $x+\Rbb_+z \subseteq K$. Let $y\in K$. We have to show that for any $t_0\in\Rbb_+$, $y+t_0z\in K$. Note that, by convexity, for all $\lambda\in\intff01$, for all $t\in \Rbb_+$ we have
		$$
		(1-\lambda)y+\lambda(x+tz)\in K
		$$
		\lc{We then define the function}{$\fct{\lambda}{\intfo{t_0}\pinf}{\intff01}{t}{\f{t_0}t}$}
		\lc{hence}{$\forall t\geq t_0\qquad \pa{1-\lambda(t)}y+\lambda(t)x+t_0z\in K$}
		\lc{We also have}{$\pa{1-\lambda(t)}y+\lambda(t)x+t_0z\underset{t\to\pinf}{\longrightarrow}y+t_0z$}
		Since $K$ is closed, we then deduce that for all $y+t_0z\in K$. Since this holds for any $y\in K$ and any $t_0\in\Rbb_+$ we end up with $z\in \rec (K)$.
	\end{itemize}
\end{proof}

When considering a closed convex set, we can look at its relative interior to get the same recession cone.

\begin{restatable}%[Folklore]
{proposition}{recConeRi}\label{prop:recConeRi}
	Let $K\subseteq E$ be a non-empty closed convex set. Then
	$
	\rec (K) = \rec(\ri (K))
	$.
\end{restatable}

\begin{proof}
	We proceed by double inclusion.
	\begin{itemize}
		\item[\CNsubset] Let $v\in\rec (K)$. Let $x\in\ri (K)$. In particular, $x\in K$. By definition, for any $\lambda\in\Rbb_+$, $x+\lambda v\in K$. Let $S=\enstq{\lambda\in\Rbb_+}{x+\lambda v\in K\setminus \ri (K)}$. We just have to show that $S=\emptyset$. Assume $S\neq\emptyset$ and consider $\mu\in S$. Let $\lambda > \mu$. Note that 
		$$
		x+\mu v = \pa{1-\f\mu\lambda}x+\f\mu\lambda\pa{x+\lambda v}
		$$
		We have two cases :
		\begin{itemize}
			\item $\underline{\lambda\in S}$, in this case, using Proposition \ref{prop:riBorder}, since $x\in\ri (K)$ and $x+\lambda v\in K\setminus\ri (K)$, we have $x+\mu v\in\ri K$,
			which is a contradiction.
			
			\item $\underline{\lambda \notin S}$, since, by Proposition \ref{prop:riNonEmpty}, $\ri (K)$ is convex, $x\in\ri (K)$ and $x+\lambda v\in\ri (K)$, we again reach $x+\mu v\in\ri (K)$, a contradiction.
		\end{itemize}
		Both cases are impossible. Therefore, $S=\emptyset$.
		
		\item[\CSsubset] Let $v\in\rec(\ri (K))$. By Proposition \ref{prop:riNonEmpty}, there is some $x\in \ri (K)$. Therefore, for all $\lambda\in\Rbb_+$, $x+\lambda v\in \ri (K)\subseteq K$. By Proposition \ref{prop:recConeClosedConvex}, we get that $v\in\rec (K)$.
	\end{itemize}
\end{proof}

\begin{remark}
	Note that if $K$ is not closed we have, thanks to Proposition \ref{prop:riNonEmpty},  $\rec(\bar K) = \rec (\ri (K))$ but we may have $\rec (K)\neq\rec(\ri (K))$.
\end{remark}

%\el{This one is never quoted. Do you need it?}
%\begin{proposition}%[Folklore]
%	Let $K\subseteq E$ be a closed convex set. Then $\rec K$ is a closed convex cone. Moreover, $K$ is bounded if and only if $\rec K=\{0\}$.
%\end{proposition}

\begin{restatable}%[Folklore]
{lemma}{linClosedCvCone}\label{lem:linClosedCvCone}
	Let $C$ be a closed convex cone in $E$. Let $u: E\to E$ be linear. Then $u(C)$ is a closed convex cone.
\end{restatable}

\begin{proof}
	By definition of a cone, 
	\centre{$C = \{0\}\cup\Rbb_+\enstq{x\in C}{\norm x{} = 1}$}
	Since $C$ is closed, $\enstq{x\in C}{\norm x{} = 1}$ is bounded and closed in a vector space of finite dimension, hence it is compact. By linearity of $u$,
	\centre{$u(C) = \{0\}\cup\Rbb_+u(\enstq{x\in C}{\norm x{} = 1})$}
	Since $u$ is linear over a vector space of finite dimension, it is continuous. Thus, the set $u(\enstq{x\in C}{\norm x{} = 1})$ is also compact, hence closed. The continuity of the norm ensures that $u(C)$ is closed. By linearity of $u$, we also get that $u(C)$ is a convex cone.
\end{proof}

\begin{restatable}%[Folklore]
{lemma}{riCVectC}
	\label{lem:riCVectC}
	Let $C$ be a non-trivial convex cone in $E$. Let $x\in\ri(C)\setminus\{0\}$ and $u\in \Vect (C)$. Then there is $\lambda\geq 0$ such that $u+\lambda x\in C$.
\end{restatable}

\begin{proof}
	If $x=u$ then $\lambda=0$ works. We then assume $x\neq u$. Since $u\in\Vect (C)$, there is $\mu\in\intoo01$ $\mu u + (1-\mu)x\in\ri (C)$. 
	Therefore, for any $\lambda\in\Rbb_+$, $\lambda\pa{\mu u + (1-\mu)x}\in\ri (C)$. In particular, for $\lambda = \f1\mu$ (which exists since $\mu\neq 0$), 
	\centre{$u+\f{1-\mu}\mu x\in\ri (C)$}
	and we indeed have $\f{1-\mu}\mu\geq 0$.
\end{proof}

%\begin{restatable}{lemma}{equalVectC}\label{lem:equalVectC}
%	Let $C$ be a non-trivial convex cone in $E$. Let $u,v\in \Vect C$. There are $x,y\in C$ such that $u+x=v+y\in C$.
%\end{restatable}

%% file: mwConvex.tex
%Minkowski-Weyl convex sets, or \textbf{MW-convex} sets, are sets that satisfy the Minkowski-Weyl property :

\begin{definition}%[MW-convex set]
	A closed convex set $K$ is said to be \textbf{MW-convex} if there is a compact convex set $K'$ such that
	$K=K'+\rec (K)$.
\end{definition}

This property comes from the Minkowski-Weyl Theorem for polyhedra :

\begin{theorem}[Minkowski-Weyl]
	\label{thm:MW}
	Let $K\subseteq \Rbb^d$. The following statements are equivalent:
	\begin{romanenumerate}
		\item $K=\{x\in \Rbb^d\mid Ax\leq b\}$ for some matrix $A\in\Rbb^{n\times d}$ and $b\in\Rbb^n$.
		
		\item There are finitely many points $x_1,\dots x_k,\in P$ and finitely many directions $v_1,\dots,v_p$ such that
		$$
		K = \convh(\{x_1,\dots,x_k\}) + \Sum{i=1}p\Rbb_+v_i 
		$$ 
	\end{romanenumerate}
\end{theorem}

Needing this property, we will assume that the sets $K$ we consider are MW-convex. Note 
that, among others, polyhedrons are MW-convex, and thus our results apply to a more 
general class of sets.

One of the main benefits of MW-convex sets is that they behave very nicely with linear projections. 
Unlike other convex sets, the projections ``commute'' with the operator $\rec$, giving a reciprocal to Lemma \ref{lem:projcRecc}.

\begin{restatable}{lemma}{recProj}\label{lem:recProj}
	Let $K\subseteq\Rbb^d$ be MW-convex. Let $\pi$ be a linear projection over $\Rbb^d$. We have
	$
	\rec(\pi(K))\subseteq \pi(\rec (K))
	$.
\end{restatable}

\begin{proof}
	Let $x\in\rec(\pi(K))$. If $x=0$ then we immediately have $x\in\pi(\rec (K))$. Therefore, we may assume $x\neq 0$.
	For $a\in \pi(K)$, we have $a+\Rbb_+x\subseteq \pi(K)$. Thus,
	$$
	\forall\lambda\in\Rbb_+\quad\exists b(\lambda)\in\Kernel\pi \qquad a+\lambda x+b(\lambda)\in K
	$$
	\lc{Let $K'$ convex compact such that}{$K = K' + \rec (K)$}
	Therefore, for all $\lambda\in\Rbb_+$ there are $a'(\lambda)\in\pi(K')$ and $x'(\lambda)\in\pi(\rec (K))$ such that 
	$$
	a+\lambda x = a'(\lambda) + x'(\lambda)
	$$
	Since $a'(\lambda)\in \pi(K')$ and that $\pi(K')$ is compact (as the continuous image of a compact), there is $a'\in \pi(K')$ and an increasing sequence $\suiten[\lambda]$ that tends to infinity such that 
	$$
	a'(\lambda_n)\tendsto n\pinf a'
	$$
	\lc{Thus}{$\lambda_nx-x'(\lambda_n) \tendsto n\pinf a'-a$}
	\lc{We then get that}{	$
		\f{x'(\lambda_n)}{\lambda_n} = x + \f{a-a'}{\lambda_n} + \petito{n\to\pinf}{\f1{\lambda_n}}
		$}
	\lc{and}{$\f{x'(\lambda_n)}{\lambda_n} \tendsto n\pinf x$}
	Also $\f{x'(\lambda_n)}{\lambda_n}\in \pi(\rec (K))$. Moreover, using Lemma \ref{lem:linClosedCvCone}, $\pi(\rec (K))$ is closed. Hence, we have $x\in \pi(\rec (K))$ what concludes the proof.  
	
\end{proof}

The converse inclusion is true for general convex sets (Lemma \ref{lem:projcRecc} in the appendices). Combining this to Lemma
\ref{lem:recProj}, we have:

\begin{corollary}\label{cor:recProj}
	Let $K\subseteq\Rbb^d$ be MW-convex. Let $\pi$ be a linear projection over $\Rbb^d$. We have
	$
	\rec(\pi(K)) = \pi(\rec (K))
	$.
\end{corollary}

\begin{restatable}{corollary}{projMW}\label{cor:projMW}
	Let $K\subseteq\Rbb^d$ be MW-convex. Let $\pi:\Rbb^d\to\Rbb^d$ be a linear projection. Then $\pi(K)$ is MW-convex.
\end{restatable}

\begin{proof}
	We write $K=K'+\rec (K)$ where $K'$ is a convex compact set. Hence, since $\pi$ is continuous (linear in a finite dimensional space), $\pi(K')$ is also compact. Moreover, by linearity of $\pi$, we get that 
	\centre{$\pi(K) = \pi(K') + \pi(\rec (K))$}
	By Lemma \ref{lem:linClosedCvCone}, $\pi(\rec (K))$ is a closed convex cone. Hence, $\pi(K)$ is closed convex as a sum of closed convex sets. By Corollary \ref{cor:recProj}, we get
	\centre{$\pi(K) = \pi(K')+\rec(\pi (K))$}
\end{proof}

%% file: accExp.tex
We consider an arbitrary Euclidean space $E$ of dimension $d\in\Nbb$. We denote
$\inner{\cdot,\cdot}$ its scalar product and $\norm{\cdot}{}$ the
associated norm.

To study the sequences of the constraint loop problem, we need to identify the 
asymptotic directions these sequences are going towards, building a form of asymptotic 
expansion of those sequences.
We thus introduce the concept of accumulation expansion. 
As sequences may point in several directions, we consider the expansion of 
a subsequence that has a single main direction.

\begin{definition}\label{def:accExp}
	Let $\suiten$ be a sequence of $E$.
	An \textbf{accumulation expansion} of $\suiten$ consists in an increasing function $\psi:\Nbb\to\Nbb$, an integer $p\in\intn 0d$, some vectors $z_1,\dots,z_{p+1}\in E$ and sequences $\suite{\alpha_{k,n}}n\Nbb$ for $k\in\intn 1p$ such that
		\begin{description}
			\item[(AE1)\label{it:normz}] $\forall k\in\intn1p\qquad \norm {z_k}{}=1$

			\item[(AE2)\label{it:zOrth}] $\forall k,k'\in\intn1p\qquad \inner{z_k,z_{k'}}=\begin{accolade}
				1 & \text{if }k=k'\\ 0 & \text{if }k\neq k'
			\end{accolade}$
			\item[(AE3)\label{it:z0Orth}] $\forall k\in\intn1p\qquad \inner{z_k,z_{p+1}}=0$

			\item[(AE4)\label{it:alplhapos}] $\forall k\in\intn1p\qquad\forall n\in\Nbb\qquad \alpha_{k,n}>0$ 
			
			\item[(AE5)\label{it:alphapinf}] $\forall k\in\intn1p\qquad \alpha_{k,n}\tendsto n\pinf\pinf$ 
			
			\item[(AE6)] $\forall m\in\intn1p\qquad \alpha_{m,n} \underset{n\to\pinf}{\sim} \norm{u_{\psi(n)} - \Sum{k=1}{m-1}\alpha_{k,n}z_k}{}$
			
			\item[(AE7)\label{it:alphao}] $\forall k\in\intn1{p-1} \qquad \alpha_{k+1,n} = \petito{n\to\pinf}{\alpha_{k,n}}$ 
			
			\item[(AE8)\label{it:projSucc}] $\forall n\in\Nbb\qquad \forall \ell\leq m\in\intn1p\qquad  \inner{z_\ell, u_{\psi(n)} - \Sum{k=1}m\alpha_{k,n}z_k}=0$
			
			\item[(AE9)\label{it:decompu}] $u_{\psi(n)} = \Sum{k=1}{p}\alpha_{k,n}z_k + z_{p+1} + \petito{n\to\pinf}1$

		\end{description}
	Abusing notations, we will say that $u_{\psi(n)} = \Sum{k=1}{p}\alpha_{k,n}z_k + z_{p+1} + \petito{n\to\pinf}1$ is an accumulation expansion of $\suiten$.
\end{definition}

\begin{definition}
	Let $u=\suiten$ be a sequence of $E$. The set $\Dcal_u$ of \textbf{principal directions} of $u$ is defined by
	$$
		\Dcal_u = \enstq{z\in E}{\begin{array}{c}
				u_{\psi(n)} = \Sum{k=1}{p}\alpha_{k,n}z_k + z_{p+1} + \petito{n\to\pinf}1\text{ is an accumulation expansion}\\
				p\geq1\qqandqq z=z_1
		\end{array}}
	$$
\end{definition}

In other words, $\Dcal_u$ is the set of directions that are in the dominant position of some accumulation expansion of $u$ such that $p\geq 1$.
It also corresponds to the dominant directions of an unbounded sequence.

For $x\in E\setminus\{0\}$ we denote $\tilde x=\f{x}{\norm{x}{}}$ the associated normalized vector.

\begin{lemma}\label{lem:dominantDirection}
	Let $\suiten$ be an unbounded sequence of $E$. 
	There exist $z\in E$ a unit vector, an increasing function $\phi:\Nbb\to\Nbb$ and a sequence $\suiten[\alpha]$ such that 
	\begin{itemize}
		\item $\forall n\in\Nbb\qquad \alpha_n>0$
		\item $\alpha_n\tendsto n\pinf\pinf$
		\item $\alpha_n\underset{n\to\pinf}{\sim}\norm{u_{\phi(n)}}{}$
		\item $u_{\phi(n)} = \alpha_nz+\petito{n\to\pinf}{\alpha_n}$
		\item $\forall n\in\Nbb\qquad u_{\phi(n)}-\alpha_nz\in z^\bot$ where $z^\bot$ means the vector subspace of $E$ orthogonal to $\Vect(\{z\})$
	\end{itemize}
\end{lemma}

\begin{proof}
	Since $\suiten$ is unbounded, we can assume that we have an increasing function $\phi:\Nbb\to\Nbb$ such that for all $n\in\Nbb$, $u_{\phi(n)}\neq0$ and $\norm{u_{\phi(n)}}{}\tendsto n\pinf\pinf$.
	Therefore the sequence $\suite{\tilde{u_{\phi(n)}}}n\Nbb$ is well defined. Moreover, as
	it is bounded by definition, up to refining $\phi$, we can assume that it converges 
	to some $z\in E$. 
	Let $\pi$ be the orthogonal projection onto $\Rbb z$. We define $\alpha_n$ to be the unique real number such that $\pi(u_{\phi(n)}) = \alpha_n z$.
	As $\tilde{u_{\phi(n)}}\tendsto n\pinf z$, we have that $\alpha_n\underset{n\to\pinf}\sim\norm{u_{\phi(n)}}{}$.
	Therefore, up to refining $\phi$, we can assume that $\alpha_n\tendsto n\pinf\pinf$ and $\alpha_n>0$. Moreover,
	we have $u_{\phi(n)} = \alpha_nz+\petito{n\to\pinf}{\alpha_n}$.
	Finally, by definition of $\pi$, for all $n\in\Nbb$, $u_{\phi(n)}-\alpha_nz\in z^\bot$.
\end{proof}

\begin{proposition}\label{prop:accExp}
	Any sequence $u=\suiten$ of $E$ admits accumulation expansions. Moreover, if $u$ is unbounded, then $\Dcal_u$ is not empty.
\end{proposition}

\begin{proof}
If $\suiten$ is bounded, then it has an accumulation point $z_1$. Hence, taking $p=0$, all 
the points are trivially true except Point \ref{it:decompu}. Taking any $\psi$ given by the 
definition of accumulation point lead to 
$u_{\psi(n)} = z_1 + \petito{n\to\pinf}1$.

Assume now that $\suiten$ is unbounded.
We proceed by induction on $d=\dim E$.
	\begin{itemize}
		\item If $d=1$, consider $z_1$ and $\suite{\alpha_{1,n}}n\Nbb$ and $\psi$ given by Lemma \ref{lem:dominantDirection}. 
		By definition, $\norm {z_1}{}=1$ and $u_{\psi(n)}-\alpha_{1,n}z_1\in z_1^\bot=\{0\}$. Taking $p=1$ and $z_2=0$ satisfies all the required properties. Moreover, $z_1\in\Dcal_u$.
		
		\item Assume the proposition holds for any Euclidean space of dimension $d-1$.
		Consider $z_1$, $\suite{\alpha_{1,n}'}n\Nbb$ and $\phi$ given by Lemma \ref{lem:dominantDirection}. 
		By definition $\norm {z_1}{}=1$ and $u_{\phi(n)}-\alpha_{1,n}'z_1\in z_1^\bot$.
		Since $z_1\neq 0$, $\dim z_1^\bot=d-1$. We can thus apply the induction hypothesis on the sequence $\suite{u_{\phi(n)}-\alpha_{1,n}'z_1}n\Nbb$ in $z_1^\bot$.
		Let $\phi'$ be the function given by the induction hypothesis. Let 
		$\psi=\phi\circ\phi'$ and $\alpha_{1,n} = \alpha_{1,\phi'(n)}'$.
		
		Every point is immediately satisfied either by the induction hypothesis or the fact that $z_1$ is orthogonal to any point in $z_1^\bot$, except for Point \ref{it:alphao}: It remains to prove that if $p\geq2$, then $\alpha_{2,n}=\petito{n\to\pinf}{\alpha_{1,n}}$.
		By induction hypothesis we know that
		\centre{$\alpha_{2,n}\underset{n\to\pinf}{\sim}\norm{u_{\psi(n)}-\alpha_{1,n}z_1}{}$}
		\lc{Moreover, by Lemma \ref{lem:dominantDirection}}{$\norm{u_{\phi(n)}-\alpha_{1,n}'z_1}{} = \petito{n\to\pinf}{\alpha_{1,n}'}$}
		Since $\suite{\alpha_{1,n}}n\Nbb$ is a subsequence of $\suite{\alpha_{1,n}'}n\Nbb$, we have
		\centre{$\alpha_{2,n}\underset{n\to\pinf}{\sim}\norm{u_{\psi(n)}-\alpha_{1,n}z_1}{} = \petito{n\to\pinf}{\alpha_{1,n}}$}
		as required.  Moreover, $z_1\in\Dcal_u$.
	\end{itemize}
\end{proof}

We now state a relation between the directions within the accumulation expansion and
the set $\rec (K)$.

\begin{restatable}{proposition}{falsetrue}\label{prop:accExpRecK}
	Let $E$ be an Euclidean space.
	Let $K\subseteq E$ be MW-convex. 
	Let $u=\suiten$ be an unbounded sequence in $K$.
	Let $u_{\phi(n)} = \Sum{k=1}{p}\alpha_{k,n}z_k + z_{p+1} + \petito{n\to\pinf}1$
	be an accumulation expansion of $\suiten$. Then, there are some positive real numbers $\pa{\beta_{k,\ell}}_{1\leq \ell < k\leq p+1}$ such that
	\centre{$
		\forall k\in\intn1p\qquad z_k +\Sum{\ell=1}{k-1} \beta_{k,\ell}z_\ell\in \rec (K)$ 
	}
	\lc{and }{$z_{p+1} +\Sum{\ell=1}{p} \beta_{p+1,\ell}z_\ell\in K$.}
\end{restatable}

\begin{proof}
	For $k\in\intn1p$, we consider $\pi_k: E\to E$ the orthogonal projection onto the vector space $\Vect((z_1,\dots, z_{k-1})^\bot)$. Let us first show that $z_k \in \pi_k(\rec (K))$.
	Let $\lambda\in\Rbb_+$ and define
	\centre{$\lambda_{k,n} = \f\lambda{\alpha_{k,n}}$}
	Note that for large enough $n$, $\lambda_{k,n}\in\intff01$. Without loss of generality, we assume $\lambda_{k,n}\in\intff01$. Then, by convexity,
	\centre{$\lambda_{k,n}u_{\phi(n)} + (1-\lambda_{k,n})u_0 \in K$}
	Moreover,
	\begin{calculs}
		& \pi_k\pa{\lambda_{k,n}u_{\phi(n)} + (1-\lambda_{k,n})u_0} &=&
		\lambda z_k + \Sum{\ell=k+1}{p}\lambda_{k,n}\alpha_{\ell,n}z_\ell \\
		&&&\quad+ \lambda_{k,n}z_{p+1} + (1-\lambda_{k,n})\pi_k(u_0) + \petito{n\to\pinf}{\lambda_{k,n}}\\
		&&\tendsto n\pinf & \lambda z_k + \pi_k(u_0)
	\end{calculs}
	Also, thanks to Corollary \ref{cor:projMW}, we have
	$\bar{\pi_k(K)} = \pi_k(K)$.
	Using now Proposition \ref{prop:recConeClosedConvex}, we then conclude that $z_k\in \rec(\pi_k(K))$.
	Finally, using Corollary \ref{cor:recProj}, 
	\centre{$z_k \in \pi_k(\rec (K))$}
	
	We now prove the proposition by induction on $k$.
	For $k=1$, our preliminary result gives in particular that $z_1\in\rec (K)$.  
	
	Assume now that $\pa{\beta_{q,\ell}}_{1\leq\ell<q<k}$ have been defined for some $k\in\intn1p$. Since $z_k\in\pi_k(\rec (K))$ as proven earlier, there are some real numbers $\suite{\gamma_{k,\ell}}\ell{\intn1{k-1}}$ such that
	\centre{$z_k+\Sum{\ell=1}{k-1}\gamma_{k,\ell}z_\ell \in\rec (K)$}
	If all the $\gamma_{k,\ell}$ are positive then fixing $\beta_{k,\ell}=\gamma_{k,\ell}$ satisfies the proposition. 
	Let $\ell\in\intn1{k-1}$ maximum such that $\gamma_{k,\ell}\leq 0$. Then, as by hypothesis
	we have that $z_\ell +\Sum{j=1}{\ell-1} \beta_{\ell,j}z_j\in\rec (K)$, we can deduce that

	\centre{$z_k+\Sum{j=1}{k-1}\gamma_{k,j}z_j + \pa{1+|\gamma_{k,\ell}|}\pa{z_\ell +\Sum{j=1}{\ell-1} \beta_{\ell,j}z_j}\in\rec (K)$} 
	\lc{Considering}{$\gamma_{k,j}' = \begin{accolade}
			\gamma_{k,j} & j>\ell\\
			1 & j=\ell\\
			\gamma_{k,j} + (1+|\gamma_{k,\ell}|)\beta_{\ell,j} & j<\ell
		\end{accolade}$}
	We end up with $z_k+\Sum{\ell=1}{k-1}\gamma_{k,\ell}'z_\ell \in\rec (K)$ with one less non-positive coefficient. Repeating this procedure until every coefficient is positive lead to a sum of the desired shape, thus establishing the induction hypothesis holds on $k$ and therefore concluding the induction.
	
	Let $\pi_{p+1}:E\to E$ the orthogonal projection on $\Vect((z_1,\dots, z_p)^\bot)$. We have
	\centre{$\pi_{p+1}\pa{u_{\phi(n)}}\tendsto n\pinf z_{p+1}$}
	\lc{By Corollary \ref{cor:projMW},}{$z_{p+1}\in\bar{\pi_{p+1}(K)} = \pi_{p+1}(K)$}
	Thus, there are some real numbers $\suite{\gamma_{p+1,\ell}}\ell{\intn1k}$ such that
	\centre{$z_{p+1}+\Sum{\ell=1}{p}\gamma_{p+1,\ell}z_\ell \in K$}
	Doing the same work as above, we can add some elements of $\rec (K)$ so that we end up with some positive $\suite{\beta_{p+1,\ell}}\ell{\intn1k}$ such that 
	\centre{$z_{p+1}+\Sum{\ell=1}{p}\gamma_{p+1,\ell}z_\ell \in K$}
\end{proof}

The two following corollaries specialise this result for some form of sequences. 

\begin{restatable}{corollary}{falsetruecor}\label{cor:accExpRecKProj}
	Let $E$ an Euclidean space.
	Let $\pi:E\to E$ be a linear projection. Let $K\subseteq E$ be MW-convex. 
	Let $u=\suiten$ be an unbounded sequence in $K$ and $x\in\Dcal_{\pi(u)}$. 
	Let 
	\centre{$\f{(\id-\pi)(u_{\phi(n)})}{\norm{\pi(u_{\phi(n)})}{}} = \Sum{k=1}{p}\alpha_{k,n}z_k + z_{p+1} + \petito{n\to\pinf}1$}
	be an accumulation expansion of $\suite{\f{(\id-\pi)\pa{u_n}}{\norm{\pi(u_n)}{}}}n\Nbb$ such that
	\centre{$\tilde{\pi(u_{\phi(n)})}\tendsto n\pinf x$} 
	Then, there are some positive real numbers $\pa{\beta_{k,\ell}}_{1\leq \ell < k\leq p+1}$ such that
	\centre{$
		\forall k\in\intn1{p+1}\qquad z_k +\Sum{\ell=1}{k-1} \beta_{k,\ell}z_\ell\in \rec (K)
		\qqandqq x+z_{p+1} +\Sum{\ell=1}{p} \beta_{p+1,\ell}z_\ell\in \rec (K)$}
\end{restatable}

\begin{proof}
	We have
	\centre{$(\id-\pi)\pa{u_{\phi(n)}} = \Sum{k=1}{p}\norm{\pi(u_{\phi(n)})}{}\alpha_{k,n}z_k + \norm{\pi(u_{\phi(n)})}{}z_{p+1} + \petito{n\to\pinf}{\norm{\pi(u_{\phi(n)})}{}}$}
	Also, provided $\tilde{u_{\phi(n)}}\tendsto n\pinf x$, we have
	\centre{$\pi\pa{u_{\phi(n)}} = \norm{\pi(u_{\phi(n)})}{}x + \petito{n\to\pinf}{\norm{\pi(u_{\phi(n)})}{}}$}
	Therefore
	\centre{$u_{\phi(n)} = \Sum{k=1}{p}\norm{\pi(u_{\phi(n)})}{}\alpha_{k,n}z_k + \norm{\pi(u_{\phi(n)})}{}\pa{x+z_{p+1}} + \petito{n\to\pinf}{\norm{\pi(u_{\phi(n)})}{}}$}
	The result is obtained by applying Proposition \ref{prop:accExpRecK} to this 
	accumulation expansion of $\suiten$.
	Note that in this case we in fact have a truncated accumulation expansion so the case $p+1$ is not the last element of an actual accumulation expansion. That is why we get $\rec (K)$ instead of $K$ even for $p+1$.
\end{proof}

\begin{restatable}{corollary}{falsetruecorr}\label{cor:accExpRecKProj2}
	Let $E$ an Euclidean space.
	Let $K\subseteq E^2$ be MW-convex. 
	Let $\pi:E\to E$ be a linear projection.
	Let $u=\suiten$ be an unbounded sequence in $E$ such that
	\centre{$\forall n\in\Nbb \qquad \pa{u_n,u_{n+1}}\in K$}
	and $x\in\Dcal_u$. 
	Let 
	\centre{$\f{u_{\phi(n)+1}}{\norm{u_{\phi(n)}}{}} = \Sum{k=1}{p}\alpha_{k,n}z_k + z_{p+1} + \petito{n\to\pinf}1$}
	be an accumulation expansion of $\suite{\f{u_{n+1}}{\norm{u_n}{}}}n\Nbb$such that
	\centre{$\tilde{u_{\phi(n)}}\tendsto n\pinf x$} 
	Then, there are some positive real numbers $\pa{\beta_{k,\ell}}_{1\leq \ell < k\leq p+1}$ such that
	\centre{$
		\forall k\in\intn1p\qquad \pa{0,z_k +\Sum{\ell=1}{k-1} \beta_{k,\ell}z_\ell}\in \rec (K)
		\qandq \pa{x,z_{p+1} +\Sum{\ell=1}{p} \beta_{p+1,\ell}z_\ell}\in \rec (K)$}
	and such that for sufficiently large $n$,
	\centre{$ \inner{\pi\pa{z_{p+1} +\Sum{\ell=1}p \beta_{p+1,\ell}z_\ell},\pi\pa{z_{p+1} +\Sum{k=1}p\alpha_{k,n}z_k}}\geq0$}
	Moreover, there is some $i\in\intn1{p+1}$ such that $\pi(z_i)\notin\Kernel(\pi)$, this inequality can be taken to be strict. 
\end{restatable}

\begin{proof}
	We first apply Corollary \ref{cor:accExpRecKProj} to the sequence $\suite{(u_n,u_{n+1})}n\Nbb$ and the projection on the first component to get
	some positive real numbers $\pa{\beta_{k,\ell}}_{1\leq \ell < k\leq p+1}$ such that
	\centre{$
		\forall k\in\intn1p\quad \pa{0,z_k +\Sum{\ell=1}{k-1} \beta_{k,\ell}z_\ell}\in \rec (K)
		\qandq \pa{x,z_{p+1} +\Sum{\ell=1}{p} \beta_{p+1,\ell}z_\ell}\in \rec (K)$}
	
	Let $k_0\in\intn1{p+1}$ minimum such that $z_k\notin\Kernel\pi$. 
	\begin{itemize}
		\item 	If there is no such $k_0$, then
		\centre{$ \inner{\pi\pa{z_{p+1} +\Sum{\ell=1}p \beta_{p+1,\ell}z_\ell},\pi\pa{ z_{p+1}+\Sum{k=1}p\alpha_{k,n}z_k }} = 0$}
		and the proof is complete.
		
		\item If $k_0=p+1$, then
		\centre{$ \inner{\pi\pa{z_{p+1} +\Sum{\ell=1}p \beta_{p+1,\ell}z_\ell},\pi\pa{z_{p+1}+\Sum{k=1}p\alpha_{k,n}z_k}} = \inner{\pi(z_{p+1}),\pi(z_{p+1})}> 0$}
		\item Otherwise, $k_0\in\intn1p$ and $\norm{\pi(z_{k_0})}{}\neq 0$. Let
		\centre{$S_n(\lambda) = 
			\inner{\pi\pa{z_{p+1} +\Sum{\ell=1}p \beta_{p+1,\ell}z_\ell} + \lambda\pi\pa{z_{k_0} +\Sum{\ell=1}{k_0-1} \beta_{k_0,\ell}z_\ell},\pi\pa{\Sum{k=1}p\alpha_{k,n}z_k + z_{p+1}}}$}
		
		We have
		
		\begin{calculs}
			& S_n(\lambda) &=& \inner{\pi(z_{p+1}) +\Sum{\ell=k_0}p \beta_{p+1,\ell}\pi(z_\ell) +
				\lambda\pi\pa{z_{k_0}},\Sum{k=k_0}p\alpha_{k,n}\pi(z_k) + \pi(z_{p+1})}\\
			&&=& \inner{\pi(z_{p+1}) +\Sum{\ell=k_0}p \beta_{p+1,\ell}\pi(z_\ell) + \lambda\pi\pa{z_{k_0}},\pi(z_{p+1})} \\&&&
			+ \Sum{k=k_0}p\alpha_{k,n}\inner{\pi(z_{p+1}) +\Sum{\ell=k_0}p \beta_{p+1,\ell}\pi(z_\ell) + \lambda\pi\pa{z_{k_0}},\pi(z_k)}\\
			
			&&=&\alpha_{k_0,n}\inner{\pi(z_{p+1}) +\Sum{\ell=k_0}p \beta_{p+1,\ell}\pi(z_\ell) + \lambda\pi\pa{z_{k_0}},\pi(z_{k_0})} 
			+\petito{n\to\pinf}{\alpha_{k_0,n}}\\
			
			&&=& \alpha_{k_0,n}\pa{\lambda\norm{\pi(z_{k_0})}{}^2 + \inner{\pi(z_{p+1}) +\Sum{\ell=k_0}p \beta_{p+1,\ell}\pi(z_\ell),\pi(z_{k_0})}}
			+\petito{n\to\pinf}{\alpha_{k_0,n}}\\
		\end{calculs}
		
		Therefore, taking any $\lambda>0$ such that
		\centre{$\lambda > - \f{\inner{\pi(z_{p+1}) +\Sum{\ell=k_0}p \beta_{p+1,\ell}\pi(z_\ell),\pi(z_{k_0})}}{\norm{\pi(z_{k_0})}{}^2}$}
		\lc{we get}{$S_n(\lambda)\tendsto n\pinf\pinf$}
		Thus, for sufficiently large $n$,
		\centre{$\inner{\pi\pa{z_{p+1} +\Sum{\ell=1}p \beta_{p+1,\ell}z_\ell} + \lambda\pi\pa{z_{k_0} +\Sum{\ell=1}{k_0-1} \beta_{k_0,\ell}z_\ell},\pi\pa{\Sum{k=1}p\alpha_{k,n}z_k + z_{p+1}}} >0$}
		Also,
		\begin{calculs}
			& \pa{x,z_{p+1} +\Sum{\ell=1}{p} \beta_{p+1,\ell}z_\ell + \lambda z_{k_0} +\lambda\Sum{\ell=1}{k_0-1} \beta_{k_0,\ell}z_\ell} &=&
			\underbrace{\pa{x,z_{p+1} +\Sum{\ell=1}{p} \beta_{p+1,\ell}z_\ell}}_{\in \rec (K)}\\
			&&&\quad+ \underbrace\lambda_{\geq 0}\underbrace{\pa{0,z_{k_0} + \Sum{\ell=1}{k_0-1} \beta_{k_0,\ell}z_\ell}}_{\in\rec (K)}\\
			
			& \pa{x,z_{p+1} +\Sum{\ell=1}{p} \beta_{p+1,\ell}z_\ell + \lambda z_{k_0} +\lambda\Sum{\ell=1}{k_0-1} \beta_{k_0,\ell}z_\ell} &\in& \rec (K)
			
		\end{calculs}
		\lc{Thus, considering }{$\beta'_{p+1,\ell} = \begin{accolade}
				\beta_{p+1,\ell} & \ell > k_0\\
				\beta_{p+1,\ell} + \lambda & \ell=k_0\\
				\beta_{p+1,\ell} + \lambda\beta_{k_0,\ell} & \ell<k_0
			\end{accolade}$}
		instead of the $\beta_{k+1,\ell}$s, we get the desired result.
	\end{itemize}
\end{proof}

%% file: overview.tex
The goal of this section is to establish Theorem~\ref{thm:dim12}. This
will be done by showing equivalence between the existence of a witness
of the form given by Definition~\ref{def:wit} and the existence of an
infinite run of a constraint loop.  The easy direction in this
argument---constructing an infinite execution from a witness---is
the purpose of Subsection~\ref{sec:suff}, Proposition \ref{prop:sufficent}.  
Actually, there is an even easier case,
namely certifying the existence of bounded infinite run, is dealt with
in Section~\ref{sec:bounded}. It states that an infinite run exists if an only if there is a fixed point.
This proof holds in any dimension and
relies on a simpler certificate. We will also reuse this result in the
specific cases of dimension 1 and 2.

The main objective in this section is to construct a witness from an
infinite execution. We provide the proof of sufficient condition in 
Subsection~\ref{sec:suff}. This will enlighten why the witness is defined
the way that it is. Subsection~\ref{sec:1D} deals with the simple
1-dimensional case, and Subsection~\ref{sec:2D} handles the
dimension-2 case, which is more challenging. Because of the difficulty of this 
proof we only provide high level explanation here. For a complete proof,
we refer to the full-version of this article or to the appendices.

%% file: bounded.tex
\begin{proposition} \label{prop:fixedPoint}
	Let $E$ be a vector space of dimension $d\in\Nbb$. Let $K\subseteq E^2$ be closed convex. Denoting $\Delta_E=\enstq{(x,x)}{x\in E}\subseteq E^2$, we have that
	$K\cap \Delta_E\neq\emptyset$ if and only if there is a bounded sequence
	$u=\suiten$ of $E$
	such that for all $n\in\Nbb$, $(u_n, u_{n+1})\in K$.
\end{proposition}

\begin{proof}\hspace{1em}
	\begin{itemize}
		\item[\CN] Let $(x,x) \in K\cap \Delta_E$. The sequence 
		constantly equal to $x$ satisfy the proposition.
		
		\item[\CS] Assume now that there exists a bounded sequence
		$\suite{u_n}n\Nbb$ 
		such that for all $n\in\Nbb$, $\pa{u_n,u_{n+1}}\in K$. Let $n\in\Nbb^*$ and define
		$x_n=\f1n\Sum{p=0}n\pa{u_p,u_p}$ and $y_n=\f1n\Sum{p=0}{n-1}\pa{u_p,u_{p+1}}$. We have
		$$
		\norm{x_n-y_n}{} = \f1n\norm{\pa{u_n,u_0}}{}
		$$
		Since the sequence $\suite {u_n}n\Nbb$ is bounded, there is a positive real number $M$ such that
		$$
		\forall n\in\Nbb^*\qquad\norm{x_n-y_n}{} \leq \f Mn
		$$
		In particular, both sequences $\suite{x_n}n{\Nbb^*}$ and $\suite{y_n}n{\Nbb^*}$ must have the same accumulation points.
		As these sequences are bounded (and since they are in a vector space of finite dimension), such a point exists. Let us denote it $x$.
		Notice that since $K$ is closed and convex, for all positive integer $n$, $y_n\in K$ and thus $x\in K$. Moreover,
		by definition, for all positive integer $n$, $x_n\in\Delta_E$. This set is again closed, thus $x\in\Delta_E$. This proves that
		$$
		x\in K\cap\Delta_E\neq\emptyset
		$$
	\end{itemize}
\end{proof}

%% file: sufficient.tex
\begin{proposition}\label{prop:sufficent}
	Let $E$ be an Euclidean space of dimension $d$.
	Let $K\subseteq E^2$ be MW-convex. If there exists a witness $\Wcal(K)$,
%	there are a linear application $M:E\to E$, $\lambda\in\intff01$ and a closed convex cone $C\subseteq E$ and $x,y,z\in E$ such that
%	\begin{romanenumerate}
%		\item \label{it:MCC}$MC\subseteq C$
%		\item \label{it:xMx}$\forall x\in C\qquad (x,Mx)\in \rec K$
%		\item \label{it:xyz}$(x,y)\in K\wedge (y,z)\in K$
%		\item \label{it:zyxC}$z-\lambda y-(1-\lambda) x\in C$
%	\end{romanenumerate}
	then, there is a sequence $\suiten\in E^\Nbb$ such that
	\centre{$\forall n\in\Nbb\qquad (u_n,u_{n+1})\in K$}
\end{proposition}

\begin{proof}
	Assume we have a witness $\Wcal(K)$. We then take $M,v,w,C$ as given by the witness
	and define the following sequence:
	\centre{$u_0=v\qqandqq u_1=w$}
	\centre{$\forall n\in\Nbb\qquad u_{n+2}-u_{n+1} = M\pa{u_{n+1}-u_n}$}
	
	Remark first that for all $n\in\Nbb$, $u_{n+1}-u_n\in C$. This can be proven by induction, noting that 
	the initialisation is given by Point~\ref{it:vwzC} and the induction step comes from
    Point~\ref{it:MCC}. 
	
	We now prove by induction that $\forall n\in\Nbb\qquad (u_n,u_{n+1})\in K$.
	\begin{itemize}
		\item By Point \ref{it:vwz}, $(u_0,u_1)\in K$.
		\item Assume that for some $n\in\Nbb$, $(u_n,u_{n+1})\in K$.
		As $u_{n+1}-u_n\in C$ as shown before, by Point \ref{it:xMx}
		\centre{
			$\pa{u_{n+1}-u_n,u_{n+2}-u_{n+1}}\in \rec (K)$	
		}
		\lc{Thus}{$\pa{u_{n+1},u_{n+2}} = (u_n,u_{n+1}) + \pa{u_{n+1}-u_n,u_{n+2}-u_{n+1}}\in K+\rec (K) = K$}
	\end{itemize}
	By the induction principle we conclude that for all $n\in\Nbb$, $(u_n,u_{n+1})\in K$.
\end{proof}

%% file: dim1.tex
We establish the main result in the one dimensional case.  Note that
we prove a slightly stronger certificate here, which is not necessary
in itself, but which we need for the 2 dimensional case.

\begin{proposition}\label{prop:dim1}
	Let $E$ be an Euclidean space of dimension $1$.
	Let $K\subseteq E^2$ be MW-convex. Let a sequence $\suiten\in
        E^\Nbb$ such that $(u_n,u_{n+1})\in K$ for all $n\in \Nbb$.
	Let $\gamma\in \Cone\pa{\Dcal_u}$ such that $(0,\gamma)\in\rec (K)$ (note that at least $\gamma=0$ works).
	Then, there are $a\in\Rbb^*$, a closed convex cone $C\subseteq E$ and $x,y\in E$ such that
	\begin{romanenumerate}
		\item\label{it:inducCDim1} $aC\subseteq C$
		\item\label{it:inducuDim1} $\forall x\in C\qquad (x,ax)\in \rec (K)$
		\item\label{it:inituDim1} $(x,y)\in K$
		\item\label{it:initCDim1} $y- x\in C$
		\item $\gamma \in C$
	\end{romanenumerate}
\end{proposition}

\begin{proof}Without loss of generality, as $E$ is an Euclidean space of dimension 1, we assume $E=\Rbb$.
	If $\suiten$ is bounded, then, by Proposition \ref{prop:fixedPoint} there exists $z\in\Rbb$ such that $(z,z)\in K$. Then $\gamma=0$ and we can select $y=x=z$, $C=\{0\}$ and $a\in\Rbb^*$ arbitrary (e.g. $1$) to produce the requested witness.

We now assume that $\suiten$ is unbounded. By Proposition \ref{prop:accExp}, it admits accumulation expansions and $\Dcal_u\neq\emptyset$. 
	The only two possible accumulation directions are $1$ and $-1$. 
	We consider three cases:
	\begin{itemize}
		\item If $\Dcal_u=\{-1,1\}$. Take $\phi_1$ and $\phi_{-1}$ such that
		$\tilde{u_{\phi_1(n)}}\tendsto n\pinf1$ and $\tilde{u_{\phi_{-1}(n)}}\tendsto n\pinf-1$.
		Up to extracting a subsequence, we have the accumulation expansions
		\centre{$\f{u_{\phi_1(n)+1}}{\norm{u_{\phi_1(n)}}{}} = \Sum{k=1}{p}\alpha_{k,n}z_k + z_{p+1} + \petito{n\to\pinf}1$}
		\lc{and}{$\f{u_{\phi_{-1}(n)+1}}{\norm{u_{\phi_{-1}(n)}}{}} = \Sum{k=1}{p'}\alpha'_{k,n}z'_k + z'_{p'+1} + \petito{n\to\pinf}1$}	
		Then, by Corollary \ref{cor:accExpRecKProj2}, there are $\alpha,\beta\in\Rbb$ such that
		\centre{$(1,\alpha)\in\rec (K)\qqandqq (-1,\beta)\in\rec (K)$}
		\lc{Let}{$\delta=\begin{accolade}
				\gamma & \text{if }\gamma\neq 0\\
				\alpha+\beta & \text{if }\gamma=0
			\end{accolade}$}
		Therefore, either $(0,\delta)=(0,\gamma)\in\rec(K)$, or $(0,\delta)=(1,\alpha)+(-1,\beta)$ and $(0,\delta)\in\rec(K)$ by conic combinations. 
		\begin{itemize}
			\item If $\delta=0$ then $\gamma=0$ and $\alpha=-\beta$.
			\begin{itemize}
				\item \lc{If $\alpha = 0$, then $\beta=0$,}{$(u_{1},u_{1}) = \underbrace{(u_0,u_{1})}_{\in K} + \underbrace{(u_{1}-u_0,0)}_{\in\rec (K)} \in K$}
				We then choose for instance $a\in\Rbb^*$, $C=\{0\}$ and $x=y=u_1$. 
				
				\item If $\alpha\neq 0 $, then we just have to take $a=\alpha$, $C=\Rbb$, $x=u_0$, $y=u_1$.
			\end{itemize}
			Note that in both these cases we trivially have $\gamma\in C$.
			
			\item If $\delta>0$ then, for large enough $n$, $n\delta+\alpha>0$. Moreover, as $\rec (K)$ is a cone, \centre{$(1,n\delta+\alpha) = \underbrace{n(0,\delta)}_{\in \rec(K)} 
											+ \underbrace{(1,\alpha)}_{\in \rec(K)}\in\rec (K)$}
			We then take $a=n\delta+\alpha>0$, $C=\Rbb_+$, $x=u_k$, $y=u_{k+1}$, for some $k$ such that $u_{k+1}-u_k+>0$. This exists since $1\in\Dcal_u$ and hence $\suiten$ is not bounded from above. Note also that since $\delta>0$ then $\gamma\geq 0$. Thus $\gamma\in C$.
			
			\item If $\delta<0$ then, for large enough $n$, $n\delta+\beta<0$. Moreover, as $\rec (K)$ is a cone, \centre{$(-1,n\delta+\beta) = \underbrace{n(0,\delta)}_{\in \rec(K)}
											+ \underbrace{(-1,\beta)}_{\in \rec(K)}\in\rec (K)$}
			We then take $a=-n\delta-\beta>0$, $C=\Rbb_-$, $x=u_k$, $y=u_{k+1}$ for some $k$ such that $u_{k+1}-u_k<0$. This exists since $-1\in\Dcal_u$ and hence $\suiten$ is not bounded from below.  Note also that since $\delta<0$ then $\gamma\leq 0$. Thus $\gamma\in C$.
			
		\end{itemize}
		
		\item If $\Dcal_u=\{1\}$, then, similarly to the first case, using Corollary \ref{cor:accExpRecKProj2}, there is some $\alpha\in\Rbb_+$ such that $(1,\alpha)\in\rec (K)$. Note also that $\gamma\geq 0$ and that $(1,\alpha+\gamma)\in\rec (K)$.
		Let $k$ such that $u_{k+1}-u_k>0$. This exists since $1\in\Dcal_u$ and hence $\suiten$ is not bounded from above.
		\begin{itemize}
			\item If $\alpha+\gamma=0$, then, $\alpha=\gamma=0$ and
			\centre{$(u_{k+1},u_{k+1}) = \underbrace{(u_k,u_{k+1})}_{\in K} + \underbrace{(u_{k+1}-u_k,0)}_{\in\rec (K)} \in K$}
			
			We then choose for instance $a\in\Rbb^*$, $C=\{0\}$ and $x=y=u_{k+1}$. 
			
			\item If $\alpha+\gamma>0$, then we just have to take $a=\alpha+\gamma$, $C=\Rbb_+$, $x=u_k$ and $y=u_{k+1}$.
		\end{itemize}
		Note that in both cases, $\gamma\in\Rbb_+=C$.
		
		\item The case $\Dcal_u=\{-1\}$ can be made similarly to the previous point.
	\end{itemize}
\end{proof}

We are now ready to prove the special case of Theorem~\ref{thm:dim12}
in which $E$ has dimension 1 (see Section~\ref{sec:intro}).  Without
loss of generality we just consider $E=\Rbb$.  The necessary condition
is given by the application of Proposition~\ref{prop:dim1} with
$\gamma=0$. The sufficient condition is given by
Proposition~\ref{prop:sufficent}.

%% file: dim2.tex
We now move to 2-dimensional Euclidean spaces and prove that the
existence of a witness as given by Definition~\ref{def:wit} is implied
by the existence of an infinite sequence. This, combined with
Proposition~\ref{prop:sufficent} will imply Theorem \ref{thm:dim12}.

For the entire section, we thus fix $E$ to be an Euclidean space of
dimension $2$, $K\subseteq E^2$ to be MW-convex and thus satisfying
$K=K'+\rec (K)$ where $K'$ is a compact convex set. We assume that there
exists a sequence $\suiten\in E^\Nbb$ such that
for all $n\in\Nbb$, $(u_n,u_{n+1})\in K$.

We start by two technical lemmas to lighten the proof of the proposition.

\begin{lemma}\label{it:defS}
Assume that
$\Dcal_u$ is not empty and for all $x\in\Cone\Dcal_u$, if $(0,x)\in\rec (K)$, then $x=0$.
Denoting $\Ccal_u=\Cone\Dcal_u$, we have that
for all $x\in\Ccal_u$, there is $s(x)\in\Ccal_u$ such that $(x,s(x))\in\rec (K)$.
\end{lemma}
\begin{proof}
	Let $x\in\Ccal_u$. By definition, we can consider $x_1,\dots,x_m\in\Dcal_u$ and $\lambda_1,\dots,\lambda_m\in\Rbb_+$ such that 
	$ x =\Sum{i=1}m \lambda_ix_i$.
	By definition of $\Dcal_u$, for $i\in\intn1m$ there is an increasing function $\phi_i:\Nbb\to\Nbb$ such that
	$\tilde{u_{\phi_i(n)}}\tendsto n\pinf x_i$.
	Using Proposition \ref{prop:accExp}, the sequence $\suite{\pa{u_{\phi_i(n)},u_{\phi_i(n)+1}}}n\Nbb$ admits an accumulation expansion
	\centre{$\pa{u_{\phi_i\circ\psi_i(n)},u_{\phi_i\circ\psi_i(n)+1}} = \Sum{k=1}{p_i}\alpha_{i,k,n}(z_{i,k,1},z_{i,k,2}) + (z_{i,p+1,1},z_{i,p+1,2}) + \petito{n\to\pinf}1 $}
	In particular, for $k\in\intn1{p_i}$ minimum such that $z_{i,k,1}\neq 0$, we have $z_{i,k,1}\in\Rbb_+^*x_i$. 
	Since the first component is not bounded, such a $k$ exists.
	Let $\mu_i >0$ such that $z_{i,k,1}=\mu_i x_i$.
	Now, applying Proposition \ref{prop:accExpRecK}, $(z_{i,1,1},z_{i,1,2})\in\rec (K)$ and $\norm{(z_{i,1,1},z_{i,1,2})}{}=1$. 
	Therefore, if $k>1$, then $z_{i,1,1}=0$ and $\norm{z_{i,1,2}}{}=1$. Hence $z_{i,1,2}\in\Dcal_u$. This contradicts the hypothesis that for all $x\in\Ccal_u$, if $(0,x)\in\rec (K)$, then $x=0$.
Thus, $k=1$.
Considering $s(x_i) = \f1{\mu_i} z_{i,1,2}$ satisfies the claim for $x_i$. Thus, defining
$s(x)=\Sum{i=1}m\lambda_i s(x_i)$
establishes the lemma.

\end{proof}

\begin{lemma}\label{it:defdelta}
Assume that $\Dcal_u$ is not empty, that for all $x\in\Cone\Dcal_u$, if $(0,x)\in\rec (K)$, then $x=0$ and for all $x\in E$, $(x,x)\not\in K$.
Denoting $\Ccal_u=\Cone\Dcal_u$, for all $x\in\Dcal_u$, there are $\delta(x)\in E$ and 
$\lambda\in\Rbb_+^*$ such that $(\delta(x),\lambda x+\delta(x))\in K\cup \rec (K)$.
\end{lemma}
\begin{proof}
Let $x\in\Dcal_u$ and the accumulation expansion
\centre{$u_{\phi(n)} = \Sum{k=1}p\alpha_{k,n}z_k + z_{p+1} + \petito{n\to\pinf}1 $} 
with $p>0$ and $z_1 = x$.
By convexity, we have
\centre{$\forall n\in\Nbb \qquad \f1{\phi(n)}\Sum{k=0}{\phi(n)-1}(u_k,u_{k+1}) \in K$}
				
Up to refining $\phi$, we can assume that we also have the accumulation expansion
				\centre{$\f1{\phi(n)}\Sum{k=0}{\phi(n)-1}(u_k,u_{k+1}) = \Sum{k=1}q\beta_{k,n}(w_{k,1},w_{k,2}) + (w_{q+1,1},w_{q+1,2}) + \petito{n\to\pinf}1 $}
Therefore
\begin{calculs}
	& \Sum{k=1}q\beta_{k,n}(w_{k,2}-w_{k,1}) + w_{q+1,2}-w_{q+1,1} 
	&=& \f1{\phi(n)}\Sum{k=0}{\phi(n)-1}(u_{k+1}-u_k) + \petito{n\to\pinf}1\\
	&&=& \f{u_{\phi(n)}-u_0}{\phi(n)} + \petito{n\to\pinf}1\\
	&&=& \Sum{k=1}p\f{\alpha_{k,n}}{\phi(n)}z_k + \petito{n\to\pinf}1	
				\end{calculs}
				
If $\suite{\f{\alpha_{1,n}}{\phi(n)}}n\Nbb$ has an accumulation point, say $\lambda$, up to refining $\phi$, we assume that it converges to it. By definition of an accumulation expansion, we then have for all $k\in\intn 1q$, $w_{k,1} = w_{k,2}$
Therefore, $w_{q+1,2}-w_{q+1,1} = \lambda x$.

				By Proposition \ref{prop:accExpRecK}, there are some positive real numbers $\gamma_1,\dots,\gamma_q$ such that
				\centre{$\Sum{k=1}q\gamma_k(w_{k,1},w_{k,2}) + (w_{q+1,1},w_{q+1,2}) \in K$}
				The difference between the two coordinates of this vector is $\lambda x$.	Since $\lambda$ is the limit of a positive sequence, $\lambda\geq 0$. 
				Also, provided that there is no $a\in E$ such that $(a,a)\in K$ by hypothesis, we have $\lambda\neq 0$. Therefore, considering $\delta(x)= \Sum{k=1}q \gamma_kw_{k,1}+w_{q+1,1}$ we get $(\delta(x),\lambda x+\delta(x))\in K$.
				
				Now if $\suite{\f{\alpha_{1,n}}{\phi(n)}}n\Nbb$ has no accumulation point. Since it is positive, we have
				\centre{$\f{\alpha_{1,n}}{\phi(n)} \tendsto n\pinf\pinf$}
				Thus, there is $k\in\intn1q$ minimum such that $w_{k,1}\neq w_{k,2}$ and for this $k$, we have
				\centre{$\beta_{k,n}(w_{k,2}-w_{k,1}) \underset{n\to\pinf}{\sim} \f{\alpha_{1,n}}{\phi(n)}x$}
				Therefore, there is $\lambda>0$ such that $w_{k,2}-w_{k,1} = \lambda x$.
				By Proposition \ref{prop:accExpRecK}, there are some positive real numbers $\gamma_1,\dots,\gamma_{k-1}$ such that
				\centre{$\Sum{\ell=1}{k-1}\gamma_\ell(w_{\ell,1},w_{\ell,2}) + (w_{k,1},w_{k,2}) \in \rec (K)$}
				The difference between the two coordinates of this vector is $\lambda x$. 
				Therefore, considering $\delta(x)= \Sum{\ell=1}{k-1} \gamma_\ell w_{\ell,1}+w_{k,1}$ we have
				$(\delta(x),\lambda x+\delta(x))\in \rec (K)$.
				\end{proof}

\begin{restatable}{proposition}{dimdeuxnec}\label{prop:dim2nec}
	There exists a witness $\Wcal(K)$.
\end{restatable}

For the detailed proof we refer to Appendix \ref{sec:detailedProof}. Here we just give an overview of the proof.

\begin{proof}[Proof sketch]
	\begin{figure}[h]
		\centre{\includegraphics[scale=.8]{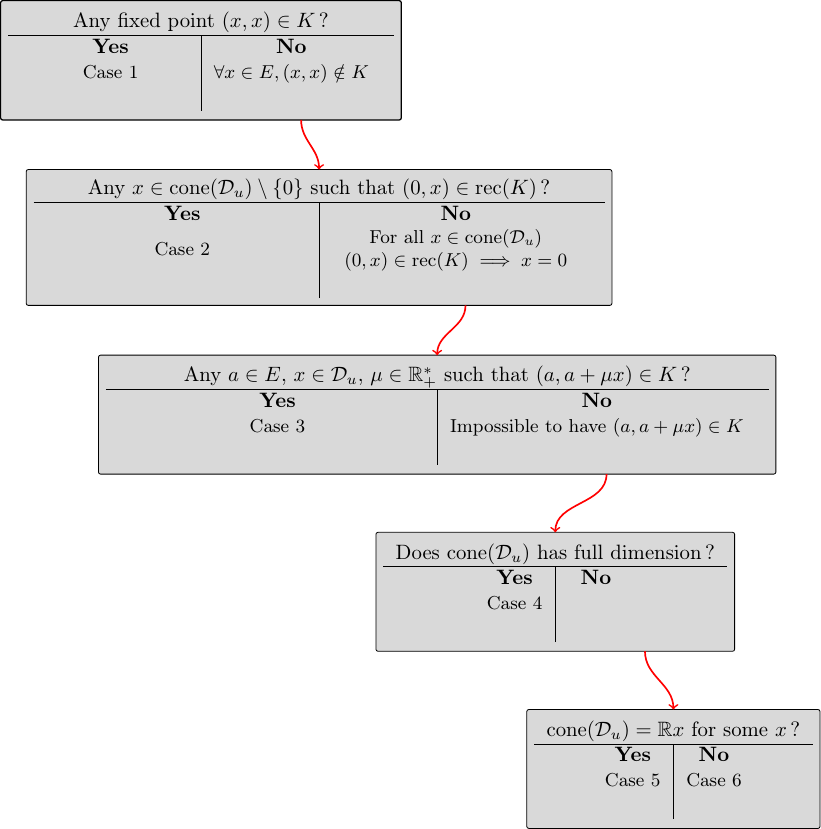}}
		\caption{The case disjunction structure}
		\label{fig:struct}
	\end{figure}
	The proof is divided into several cases under the structure described in Figure~\ref{fig:struct}.
	Among all these cases, Case 6 is by far the most difficult, followed by Cases 2 and 5, then Case 4 (quite easy) and finally the almost trivial Cases 1 and 3.
	In this proof we denote $\Ccal_u=\Cone(\Dcal_u)$.
	
	\begin{itemize}
		\item \textbf{Case 1:} There is a fixed point $(x,x)$ in $K$. In this case we just need to take $v=w=x$, $M$ arbitrary and $C=\{0\}$ to get a Witness. This just leads to a constant sequence.
		
		\item \textbf{Case 2:} No fixed point but there is $x\in\Ccal_u\setminus\{0\}$ such that $(0,x)\in \rec(K)$. In this case we are going to try to make use of Proposition \ref{prop:dim1}.
		Let $\pi:E\to E$ be the orthogonal projection onto $x^\bot$ and let $\hat\pi:E^2\to E^2$ be such that
		\centre{$\forall e,f\in E, \hat\pi(e,f)=\pa{\pi(e),\pi(f)}$} 
		Assume that we have found some $x'$ such that $(x,x')\in\rec(K)$. We then can write $x'=\gamma x+y$ for some $y$ orthogonal to $x$ and some $\gamma\in\Rbb$. Then we have $\hat\pi(x,x')=(0,y)$. Also $(0,y)\in \hat\pi(\rec(K))=\rec(\hat\pi(K))$. Thus if we can build $x'$ such that $y\in\Cone(\Dcal_{\pi(u)})$, we would be allowed to apply Proposition \ref{prop:dim1}. This requires some work. The idea is to write $x=\Sum{i=1}na_ix_i$ with $x_i\in\Dcal_u$ and $a_i\geq0$ then apply Corollary \ref{cor:accExpRecKProj2} for all $i\in\intn1n$ (details in Section \ref{sec:detailedProof}). Assume this is done. There are $a\in\Rbb^*$, a closed convex cone $C\subseteq x^\bot$ and $v,w\in x^\bot$ such that
		\begin{itemize}
			\item $aC\subseteq C$
			\item $\forall c\in C\qquad (c,ac)\in \rec (\hat\pi (K))$
			\item $(v,w)\in \hat\pi (K)$
			\item $w-v\in C$
			\item $y\in C$
		\end{itemize} 
		Again, since $\rec (\hat\pi K) = \hat\pi(\rec (K))$, for all $\xi\in C$, there are $b_\xi,c_\xi\in\Rbb$ such that
		$\pa{b_\xi x+\xi,c_\xi x+a\xi}\in\rec (K)$.
		For all $\xi\in C$, we fix $b_\xi$ and $c_\xi$ such that $(|b_\xi|,|c_\xi|)$ is minimal for the lexicographic order and among all these possibilities, such that $(b_\xi,c_\xi)$ is maximal for the lexicographic order. 
		We denote 
		\centre{$\gamma_0(\xi)=\max(1,\gamma,b_\xi)\qqandqq 
			\gamma_1(\xi)=\max(|a|,b_{a\xi},c_\xi)$}
		\lc{and for $n\geq1$,}{$\gamma_{2n}(\xi) = \max\pa{a^{2n},a^{2n}b_\xi,a^{2(n-1)}c_{a\xi}}$} 
		
		\centre{$\gamma_{2n+1}(\xi)=\max(|a|^{2n+1},a^{2n}b_{a\xi},a^{2n}c_\xi)$}
		\lc{For $n\in\Nbb$, let}{$\chi_n(\xi)=\gamma_n(\xi)x+a^n\xi$}
		\lc{and}{$b_{n,\xi}' = \begin{accolade}
				a^{2n}b_\xi & n\in2\Nbb\\
				a^{2n}b_{a\xi} & n\in2\Nbb+1
			\end{accolade}$}
		We some algebraic manipulations and intensively using that $(0,x)\in\rec(K)$ to add missing weight on $x$ in the second component, we get
		\lc{}{$\forall n\in\Nbb\qquad  \pa{\chi_n(\xi), \chi_{n+1}(\xi) + \pa{\gamma_n(\xi)-b_{n,\xi}'}\chi_0\pa{y}} \in\rec (K)$}
		
		Recalling that $w-v\in C$ we define $C'=\Rbb_+x + \Sumin n\Nbb \pa{\Rbb_+\chi_n(w-v)+\Rbbplus\chi_n(y)}$
		This is the cone we want to use. It is finitely generated. Ce can also see that it cannot contain line. Since all such two-dimensional cones are generated by at most two vectors we can find such generating vectors. $M$ will just be a matrix defined thanks to its behavior on these vectors and $C'$  is defined to get stability. Finally, up to add some component on $x$ again we can get our starting conditions thanks to $v$ and $w$ (See details in Appendix \ref{sec:detailedProof}).

		\item \textbf{Case 3:} No fixed point or $x\in\Ccal_u\setminus\{0\}$ such that $(0,x)\in \rec(K)$. However there are $a\in E$, $x\in\Dcal_u$ and $\mu\in\Rbbplusstar$ such that $(a,a+\mu x)\in K$. This means that their is a principle direction of $u$  along which  it is possible to take a first step.
		In this case, we select $C=\Ccal_u$.
		$C$ is a non empty closed convex cone of $\Rbb^2$, thus, there are two vectors $x_1,x_2\in C\setminus\{0\}$ such that either $C=\Rbb x_1+\Rbb_+x_2$ or $C=\Rbb_+x_1+\Rbb_+x_2$ or $C=\Rbb x_1+\Rbb x_2$. 
		Let $I\subseteq\{1,2\}$, $I\neq\emptyset$ the largest set such that $\suite{x_i}iI$ is a free family. 
		Using the function $s$ defined by Lemma~\ref{it:defS}, we define $M$ such that $Mx_i=s(x_i)$ for all $i\in I$.
		Noting that since, for $i\in I$, $-x_i\in C$, $(0,s(x_i)+s(-x_i))\in \rec (K)$, we have that $s(-x_i)=-s(x_i)$, this choice of $M$ satisfies Points \ref{it:MCC} and \ref{it:xMx}.
		We now choose $v=a$ and $w=a+\mu x$. By assumption, $(v,w)\in K$.
		Also, $w-v = \mu x\in\Ccal_u=C$. $C, v$ and $w$ thus satisfy Points \ref{it:vwz} and \ref{it:vwzC}.
		
		\item \textbf{Case 4:} No fixed point, $x\in\Ccal_u\setminus\{0\}$ such that $(0,x)\in \rec(K)$ or $a\in E$, $x\in\Dcal_u$, $\mu\in\Rbbplusstar$ such that $(a,a+\mu x)\in K$. However $\Dcal_u$ spans the entire space $E$.
		Given that, take $(a,b)\in K$. Using Lemma \ref{lem:riCVectC} there is $\lambda\geq0$ such that $y:=b-a+\lambda x\in \Ccal_u$. 
		Let $v=a+\lambda\delta(x)$ and $w=b+\lambda(x+\delta(x))$ with $\delta$ given by Lemma~\ref{it:defdelta}. We then have $(u,v)\in K$. 
		Let $C=\bar{\Cone\enstq{s^k(y)}{k\in\Nbb}}$ with $s$ being the function defined in Lemma~\ref{it:defS}. $C$ is a closed convex cone in a 2-dimensional vector space, therefore there are vectors $\zeta_1,\zeta_2\in C\setminus\{0\}$ such that
		\centre{$C\in\{\Rbb_+\zeta_1+\Rbb_+\zeta_2,\Rbb\zeta_1+\Rbb_+\zeta_2,\Rbb_+\zeta_1+\Rbb\zeta_2,\Rbb\zeta_1+\Rbb\zeta_2\}$}
		Let $\suite{\zeta_{i,n}}n\Nbb$ be a sequence in $\Cone\enstq{s^k(y)}{k\in\Nbb}$ such that 
		\centre{$\zeta_{i,n}\tendsto n\pinf\zeta_i$}
		If $\suite{s\pa{\zeta_{i,n}}}n\Nbb$ is unbounded then Proposition \ref{prop:accExpRecK} ensures that there is some $\zeta_i'\in\Dcal_{\suite{s\pa{\zeta_{i,n}}}n\Nbb}$ such that $(0,\zeta_i')\in\rec (K)$ and $\zeta_i'\in\Ccal_u$. This is impossible by assumption on $\Ccal_u$. Therefore, it is bounded and we have  an accumulation point $\zeta_i'\in C$. Since $\rec (K)$ is closed, we also have $(\zeta_i,\zeta_i')\in\rec (K)$. Let $I\subseteq\{1,2\}$ maximal such that $\suite{\zeta_i}iI$ is a free family. Let $M$ be a matrix such that 
		\centre{$\forall i\in I\qquad M\zeta_i=\zeta_i'$}
		
		\item \textbf{Case 5:} No fixed point, $x\in\Ccal_u\setminus\{0\}$ such that $(0,x)\in \rec(K)$ or $a\in E$, $x\in\Dcal_u$, $\mu\in\Rbbplusstar$ such that $(a,a+\mu x)\in K$ and $\Ccal_u$ is a line $\Ccal_u=\Rbb x$.
		This case uses the induction hypothesis (Proposition \ref{prop:dim1}) and similar techniques as in Case 2. The main change here is that we use the function $s$ defined by Lemma \ref{it:defS}. Here $s(x)$ will have to be collinear with $x$. In stead of adding multiples of $(0,x)$, we have access to some $(x,\gamma x)\in \rec K$ and are allowed negative coefficients which makes the case relatively easy. See details in Appendix \ref{sec:detailedProof}. 
		
		\item \textbf{Case 6:} No fixed point, $x\in\Ccal_u\setminus\{0\}$ such that $(0,x)\in \rec(K)$ or $a\in E$, $x\in\Dcal_u$, $\mu\in\Rbbplusstar$ such that $(a,a+\mu x)\in K$ and $\Ccal_u = \Rbbplus x$ for some $x$.
		Let $y\in x^\bot$ such that $\norm y{}=1$.
		The main goal of this case is to find  $a,b\geq 0$ and $c,d\in\Rbb $ such that
		\centre{$(x,ax)\in\rec (K)\qqandqq (dx+y,cx+by)\in\rec (K)\qqandqq c\geq db$}
		This can be achieved by a very careful look at the asymptotic behavior of the $\suite{u_n}n\Nbb$ and more precisely its components along $x$ and $y$. Namely, the component along $x$ must blow up significantly faster than the one along $y$.
		This is where the difficulty of this case lies. We refer to Appendix \ref{sec:detailedProof} for the details.   
		This naturally leads to choose $C$ and $M$ such that:
		\centre{$C=\Rbb_+ x + \Rbb_+(dx+y)\qqandqq Mx=ax\qqandqq M(dx+y) = cx+by$} 
		immediately satisfying \ref{it:MCC} and \ref{it:xMx}.
		With the same technics we can show that there is some $n\in\Nbb$ such that 
		\centre{$\inner{u_{n+1}-u_n,x}\geq d\inner{u_{n+1}-u_n,y}$}
		Then considering $v=u_n$ and $w=u_{n+1}$.
		\begin{calculs}
			& w-v &=& u_{n+1}-u_n = \inner{u_{n+1}-u_n,x}x + \inner{u_{n+1}-u_n,y}y\\
			&&=& \pa{\inner{u_{n+1}-u_n,x}-d\inner{u_{n+1}-u_n,y}}x + \inner{u_{n+1}-u_n,y}(dx+y)\in C
		\end{calculs}
		Hence, Points \ref{it:vwz} and \ref{it:vwzC} are satisfied by $C,v,w$.
	\end{itemize} 
\end{proof}

%% file: proofWitness.tex
We consider the same context as in Subsection \ref{sec:2D}.

\dimdeuxnec*

\begin{proof}
	We prove this result through a succession of case refinements, each producing a witness.

	Assume first that there exists $x\in E$ such that $(x,x)\in K$. 
	Then we can build a witness by choosing $v=w=x$, $M$ arbitrary and $C=\{0\}$.

	We now assume that for all $x\in E$ we have that $(x,x)\not\in K$.
	By Proposition \ref{prop:fixedPoint} and the previous assumption, $\suiten$ is unbounded.
	Therefore, $\Dcal_u\neq\emptyset$.
	We consider two cases, depending on whether there exists $x\in\Cone\pa{\Dcal_u}\setminus\{0\}$ such that $(0,x)\in \rec (K)$.

	\begin{itemize}
		\item First case: there exists $x\in\Cone\pa{\Dcal_u}\setminus\{0\}$ such that $(0,x)\in \rec (K)$. 
		We then write $x=\Sum{i=1}ma_ix_i$ with $m\in\Nbb^*$, $a_i\in\Rbb^*_+$ and $x_i\in\Dcal_u$. 
		For $i\in\intn1m$ we let $\phi_i$ be such that $\tilde{u_{\phi_i}(n)}\tendsto n\pinf x_i$.
		Let $\pi:E\to E$ be the orthogonal projection onto $x^\bot$ and let $\hat\pi:E^2\to E^2$ be such that
		\centre{$\forall e,f\in E, \hat\pi(e,f)=\pa{\pi(e),\pi(f)}$}

		Applying Proposition \ref{prop:accExp} there is an accumulation expansion 
		\centre{$\f{u_{\phi_i(n)+1}}{\norm{u_{\phi_i(n)}}{}} = \Sum{k=1}{p_i}\alpha_{i,k,n}z_{i,k} + z_{i,p_i+1} + \petito{n\to\pinf}1$}
		with $p_i\in\intn02$. Applying Corollary \ref{cor:accExpRecKProj2}, we know that there are some positive real numbers $\pa{\beta_{i,k,\ell}}_{1\leq \ell < k\leq p+1}$ such that
		\centre{$
			\forall k\in\intn1{p_i}\qquad \pa{0,z_{i,k} +\Sum{\ell=1}{k-1} \beta_{i,k,\ell}z_{i,\ell}}\in \rec (K)$}
		\lc{and}{$\pa{x_i,z_{i,p_i+1} +\Sum{\ell=1}{p_i} \beta_{i,p_i+1,\ell}z_{i,\ell}}\in \rec (K)$}
		and such that for sufficiently large $n$,
		\centre{$ \inner{\pi\pa{z_{i,p_i+1} +\Sum{\ell=1}{p_i} \beta_{i,p_i+1,\ell}z_{i,\ell}},\pi\pa{\Sum{k=1}{p_i}\alpha_{i,k,n}z_{i,k} + z_{i,p_i+1}}}\geq0$}
		the inequality being strict if there is some $j\in\intn1{p_i+1}$ such that some $z_{i,j}\notin\Kernel\pi$.
		Let $x'=\Sum{i=1}m a_i \pa{z_{i,p_i+1} +\Sum{\ell=1}{p_i} \beta_{i,p_i+1,\ell}z_{i,\ell}}$. In particular, we have $(x,x')\in\rec (K)$.
		We write $x'=\gamma x+y$  with $\inner{x,y}=0$.
		Let us show that $y\in\Cone\Dcal_{\pi(u)}$ and that $(0,y)\in\rec (\hat\pi K)$
		in order to apply Proposition~\ref{prop:dim1} on the sequence projected by $\pi$ with $y$ used as the $\gamma$ in the proposition.
		If $y=0$ then trivially $y\in\Cone\Dcal_{\pi(u)}$ and $(0,y)\in\rec (\hat\pi K)$. 
		Otherwise, since all the $a_i$ are positive, there exists some $i$ such that
		\centre{$\pi\pa{z_{i,p_i+1} +\Sum{\ell=1}{p_i} \beta_{i,p_i+1,\ell}z_{i,\ell}} \in\Rbbplusstar y$}
		In particular, there is some $j\in\intn1{p_i+1}$ such that some $z_{i,j}\notin\Kernel\pi$. Thus, 
		for sufficiently large $n$,
		\centre{$ \inner{\pi\pa{z_{i,p_i+1} +\Sum{\ell=1}{p_i} \beta_{i,p_i+1,\ell}z_{i,\ell}},\pi\pa{\Sum{k=1}{p_i}\alpha_{i,k,n}z_{i,k} + z_{i,p_i+1}}}>0$}
		\lc{Also,}{$ \pi\pa{z_{i,p_i+1}+\Sum{k=1}{p_i}\alpha_{i,k,n}z_{i,k} }\in x^\bot = \Rbb y$}
		Provided that the scalar product between the above elements is positive, we then get that
		\centre{$\pi\pa{z_{i,p_i+1}+\Sum{k=1}{p_i}\alpha_{i,k,n}z_{i,k} } \in\Rbbplusstar y$}
		\lcr{Say we have}{$\pi\pa{z_{i,p_i+1}+\Sum{k=1}{p_i}\alpha_{i,k,n}z_{i,k}} = \Lambda_{i,n}y$}{(with $\Lambda_{i,n}>0$)}
		If there is $k\in\intn{1}{p_i}$ such that $\pi(z_{i,k})\neq 0$ then $\Lambda_{i,n}\tendsto n\pinf\pinf$. Otherwise it is constant.
		In both cases, for $n$ sufficiently large, it is bounded from below by some $\Lambda_i>0$. Thus
		\begin{calculs}
			&\pi\pa{u_{\phi_i(n)+1}} &=& \norm{u_{\phi_i(n)}}{}\pi\pa{\Sum{k=1}{p_i}\alpha_{i,k,n}z_{i,k} + z_{i,p_i+1}} + \petito{n\to\pinf}{\norm{u_{\phi_i(n)}}{}}\\
			&&=& \norm{u_{\phi_i(n)}}{}\Lambda_{i,n}y+\petito{n\to\pinf}{\norm{u_{\phi_i(n)}}{}}\\
			&&=& \norm{u_{\phi_i(n)}}{}\Lambda_{i,n}y+\petito{n\to\pinf}{\norm{u_{\phi_i(n)}}{}\Lambda_{i,n}}
		\end{calculs}
		
		Therefore, any accumulation expansion extracted from the above expression will stand as a witness for $y\in\Rbbplusstar\Dcal_{\pi(u)}\subseteq\Cone\Dcal_{\pi(u)}$.
		Moreover, by Corollary \ref{cor:recProj} we have $\rec (\hat\pi K) = \hat\pi(\rec (K))$. Thus
		since $\hat\pi(x,x')=(0,y)$ and $(x,x')\in\rec (K)$, we have $(0,y)\in\hat\pi(\rec (K)) = \rec (\hat\pi K)$.	
		As $\dim x^\bot = 1$, we can apply Proposition \ref{prop:dim1}:
		there are $a\in\Rbb^*$, a closed convex cone $C\subseteq x^\bot$ and $v,w\in x^\bot$ such that
		\begin{itemize}
			\item $aC\subseteq C$
			\item $\forall c\in C\qquad (c,ac)\in \rec (\hat\pi (K))$
			\item $(v,w)\in \hat\pi (K)$
			\item $w-v\in C$
			\item $y\in C$
		\end{itemize}
		
		Again, since $\rec (\hat\pi K) = \hat\pi(\rec (K))$, for all $\xi\in C$, there are $b_\xi,c_\xi\in\Rbb$ such that
		$\pa{b_\xi x+\xi,c_\xi x+a\xi}\in\rec (K)$.
		For all $\xi\in C$, we fix $b_\xi$ and $c_\xi$ such that $(|b_\xi|,|c_\xi|)$ is minimal for the lexicographic order and among all these possibilities, such that $(b_\xi,c_\xi)$ is maximal for the lexicographic order. 
		We denote 
		\centre{$\gamma_0(\xi)=\max(1,\gamma,b_\xi)\qqandqq 
			\gamma_1(\xi)=\max(|a|,b_{a\xi},c_\xi)$}
		\lc{and for $n\geq1$,}{$\gamma_{2n}(\xi) = \max\pa{a^{2n},a^{2n}b_\xi,a^{2(n-1)}c_{a\xi}}$} 
		
		\lc{and}{$\gamma_{2n+1}(\xi)=\max(|a|^{2n+1},a^{2n}b_{a\xi},a^{2n}c_\xi)$}
		
		Using that $(x,\gamma x + y)\in\rec (K)$ and $(0,x)\in\rec (K)$, we get the following
		\begin{itemize}
			\item $\pa{x,\gamma_0(y)x+y}\in\rec (K)$
			\item For all $\xi\in C$ and for all $n\in\Nbb$
			\centre{$\pa{\gamma_{2n}(\xi)x+a^{2n}\xi,\gamma_{2n+1}(\xi)x+a^{2n+1}\xi + \pa{\gamma_{2n}(\xi)-a^{2n}b_\xi}\pa{\gamma_0\pa{y}x+y}}\in\rec (K)$}
			and
			\centre{\small$\pa{\gamma_{2n+1}(\xi)x+a^{2n+1}\xi,\gamma_{2(n+1)}(\xi)x+a^{2(n+1)}\xi + \pa{\gamma_{2n+1}(\xi)-a^{2n}b_{a\xi}}\pa{\gamma_0\pa{y}x+y}}\in\rec (K)$}
		\end{itemize}
		\lc{For $n\in\Nbb$, let}{$\chi_n(\xi)=\gamma_n(\xi)x+a^n\xi$}
		\lc{and}{$b_{n,\xi}' = \begin{accolade}
				a^{2n}b_\xi & n\in2\Nbb\\
				a^{2n}b_{a\xi} & n\in2\Nbb+1
			\end{accolade}$}
		Then, we can write instead,
		\lcr{}{$\forall n\in\Nbb\qquad  \pa{\chi_n(\xi), \chi_{n+1}(\xi) + \pa{\gamma_n(\xi)-b_{n,\xi}'}\chi_0\pa{y}} \in\rec (K)$}{\numero\ast}
		Recalling that $w-v\in C$ we define $C'=\Rbb_+x + \Sumin n\Nbb \pa{\Rbb_+\chi_n(w-v)+\Rbbplus\chi_n(y)}$
		\lc{Noticing that}{$\forall \xi\in C\quad \forall n\in\Nbb^*\qquad \chi_{n+2}(\xi)=a^2\chi_n(\xi)$,}
		\lc{we can rewrite $C'$ as}{$C'=\Rbb_+x + \Sum{k=0}2 \pa{\Rbb_+\chi_n(w-v)+\Rbbplus\chi_n(y)}$.}		
		Moreover, as $(v,w)\in \hat\pi (K)$, there are $b,c\in\Rbb$ such that $\pa{bx+v, cx+w} \in K$.
		As $(0,x)\in\rec (K)$, we can assume without loss of generality that $c\geq b+\gamma_0(w-v)$.
		Thus
		\begin{calculs}
			&cx+w-bx-v&=& \underbrace{(c-b-\gamma_0(w-v))x}_{\in C'} + \underbrace{\chi_0(w-v)}_{\in C'}\in C'
		\end{calculs}
		This means that Points \ref{it:vwz} and \ref{it:vwzC} are satisfied by $C'$, $bx+v$ and $cx+w$. 
		We now need to define the matrix $M$. 
		Since for all $\xi\in C$ and $n\in\Nbb$, $\gamma_n(\xi)>0$, every $\xi'\in C'$ 
		satisfies $\langle\xi',x\rangle\geq 0$.
		Thus, $C'$ is salient (\emph{i.e.} if $\xi'\in C'$ and $-\xi'\in C'$, then $\xi'=0$).
		As a salient finitely generated convex cone of $\Rbb^2$, $C'$ is generated by 
		at most two of its generating vectors. Thus there are 
		\centre{$\zeta_1,\zeta_2\in \{x\} \cup \enstq{\chi_k(w-v)}{k\in\intn02} \cup \enstq{\chi_k\pa{y}}{k\in\intn02}$}
		\lc{such that}{$C'=\Rbb_+\zeta_1+\Rbb_+\zeta_2$.}
		By the fact that $\pa{x,\chi_0\pa{y}}\in\rec (K)$ and by Statement $\numero\ast$, there are $\zeta_1',\zeta_2'\in C'$ such that 
		\centre{$(\zeta_1,\zeta_1')\in\rec (K) \qqandqq \pa{\zeta_2,\zeta_2'}\in\rec (K)$}	
		Since $x\in C'\setminus\{0\}$, at least one of the $\zeta_i$ is not zero. 
		Let $I\subseteq\{1,2\}$, $I\neq\emptyset$ the largest set such that $\suite{\zeta_i}iI$ is a free family. 
		Taking $M$ such that $M\zeta_i=\zeta_i'$ for all $i\in I$, we define $M$ satisfying Points \ref{it:MCC} and \ref{it:xMx}.
		
		\item Second case: $\Dcal_u$ is not empty and for all $x\in\Cone\Dcal_u$, if $(0,x)\in\rec (K)$, then $x=0$.
		Denote $\Ccal_u=\Cone\Dcal_u$ and $E_u=\Vect(\Ccal_u)=\Vect(\Dcal_u)$.
		We split again the proof into several cases.
		\begin{alphaenumerate}
			\item \label{it:uxK}
			Assume first that there is $a\in E$, $x\in\Dcal_u$ and $\mu\in\Rbbplusstar$ such that $(a,a+\mu x)\in K$.
			We select $C=\Ccal_u$.
			$C$ is a non empty closed convex cone of $\Rbb^2$, thus, there are two vectors $x_1,x_2\in C\setminus\{0\}$ such that either $C=\Rbb x_1+\Rbb_+x_2$ or $C=\Rbb_+x_1+\Rbb_+x_2$ or $C=\Rbb x_1+\Rbb x_2$. 
			Let $I\subseteq\{1,2\}$, $I\neq\emptyset$ the largest set such that $\suite{x_i}iI$ is a free family. 
			Using the function $s$ defined by Lemma~\ref{it:defS}, we define $M$ such that $Mx_i=s(x_i)$ for all $i\in I$.
			Noting that since, for $i\in I$, $-x_i\in C$, $(0,s(x_i)+s(-x_i))\in \rec (K)$, we have that $s(-x_i)=-s(x_i)$, this choice of $M$ satisfies Points \ref{it:MCC} and \ref{it:xMx}.
			We now choose $v=a$ and $w=a+\mu x$. By assumption, $(v,w)\in K$.
			Also, $w-v = \mu x\in\Ccal_u=C$. $C, v$ and $w$ thus satisfy Points \ref{it:vwz} and \ref{it:vwzC}.
			
			\item \label{it:appelInduc} We now tackle the case where there is no $a\in E$, $x\in\Dcal_u$ and $\mu>0$ such that $(a,a+\mu x)\in K$.
			Using Lemma~\ref{it:defdelta}, for all $x\in\Dcal_u$, there are $\lambda > 0$ and $\delta(x)\in E$ such that $(\delta(x),\lambda x+\delta(x))\in K \cup\rec (K)$. 
			By the initial assumption of this case, we cannot have $(\delta(x),\lambda x+\delta(x))\in K$, thus $(\delta(x),\lambda x+\delta(x))\in \rec (K)$.				
			As $\rec (K)$ is a cone, we can divide by $\lambda$ and have that $(\delta(x), x+\delta(x))\in\rec (K)$. The function $\delta$ can then be extended to $\Ccal_u$ with the same property using conic combinations. Therefore,
			\centre{$\forall x\in\Ccal_u \qquad (\delta(x), x+\delta(x))\in\rec (K)$.} 
			We can now strengthen the initial assumption of this case by assuming that there is no $a\in E$, $x\in\Dcal_u$ and $\mu\in\Rbb$ (instead of $\mu>0$) such that $(a,a+\mu x)\in K$. Indeed, if there were such elements, then by assumption $\mu\leq 0$ and we would have
			\centre{$(a+(1-\mu)\delta(x),a+\mu x+(1-\mu)(x+\delta(x))) = (a+(1-\mu)\delta(x),(a+(1-\mu)\delta(x))+x) \in K$}
			which is a contradiction.

			For the remaining of the proof we fix $x\in\ri(\Ccal_u)$. 
			We can assume $x\neq0$ since if $0\in\ri(\Ccal_u)$ then $E_u=\Ccal_u=\ri(\Ccal_u)$ and 
			$\ri(\Ccal_u)\setminus \{0\}\neq \emptyset$, thus one could select a non-zero value for $x$.
			Since $x\in E_u$, $x\neq 0$ and $\dim E=2$, we then have either $E_u=E$ or $E_u=\Vect (x)$.
			We treat separately the cases $E_u=E$, $E_u=\Ccal_u=\Vect (x)$ and 
			$E_u=\Vect (x)$ but $\Ccal_u = \Rbb_+ x$.
			\begin{romanenumerate}
				\item Consider first the case where $E_u=E$. 
				In this case take $(a,b)\in K$. Using Lemma \ref{lem:riCVectC} there is $\lambda\geq0$ such that $y:=b-a+\lambda x\in \Ccal_u$. 
				Let $v=a+\lambda\delta(x)$ and $w=b+\lambda(x+\delta(x))$ with $\delta$ given by Lemma~\ref{it:defdelta}. We then have $(u,v)\in K$. 
				Let $C=\bar{\Cone\enstq{s^k(y)}{k\in\Nbb}}$ with $s$ being the function defined in Lemma~\ref{it:defS}. $C$ is a closed convex cone in a 2-dimensional vector space, therefore there are vectors $\zeta_1,\zeta_2\in C\setminus\{0\}$ such that
				\centre{$C\in\{\Rbb_+\zeta_1+\Rbb_+\zeta_2,\Rbb\zeta_1+\Rbb_+\zeta_2,\Rbb_+\zeta_1+\Rbb\zeta_2,\Rbb\zeta_1+\Rbb\zeta_2\}$}
				Let $\suite{\zeta_{i,n}}n\Nbb$ be a sequence in $\Cone\enstq{s^k(y)}{k\in\Nbb}$ such that 
				\centre{$\zeta_{i,n}\tendsto n\pinf\zeta_i$}
				If $\suite{s\pa{\zeta_{i,n}}}n\Nbb$ is unbounded then Proposition \ref{prop:accExpRecK} ensures that there is some $\zeta_i'\in\Dcal_{\suite{s\pa{\zeta_{i,n}}}n\Nbb}$ such that $(0,\zeta_i')\in\rec (K)$ and $\zeta_i'\in\Ccal_u$. This is impossible by assumption on $\Ccal_u$. Therefore, it is bounded and we have  an accumulation point $\zeta_i'\in C$. Since $\rec (K)$ is closed, we also have $(\zeta_i,\zeta_i')\in\rec (K)$. Let $I\subseteq\{1,2\}$ maximal such that $\suite{\zeta_i}iI$ is a free family. Let $M$ be a matrix such that 
				\centre{$\forall i\in I\qquad M\zeta_i=\zeta_i'$}

				Hence, $M$ and $C$ satisfy Points \ref{it:MCC} and \ref{it:xMx}.
				Now notice that 
				\centre{$w-v = z-v = y \in C$}
				Hence, $v,w,C$ satisfy Points \ref{it:vwz} and \ref{it:vwzC}.

				\item Consider now the case where $E_u=\Ccal_u=\Vect (x)$. 
				Let $\pi:E\to E$ the orthogonal projection onto $E_u^\bot$. Note that $\dim E_u^\bot = 1$.
				Let $\hat\pi:E^2\to E^2$ such that
				\centre{$\forall e,f\in E, \hat\pi(e,f)=\pa{\pi(e),\pi(f)}$} 
				By Corollary \ref{cor:recProj} we have $\rec (\hat\pi K) = \hat\pi(\rec (K))$, hence applying Proposition \ref{prop:dim1}, there are $M'\in\Rbb^*$, a closed convex cone $C'\subseteq E_u^\bot$ and $v',w'\in E_u^\bot$ such that
				\begin{itemize}
					\item $M'C'\subseteq C'$
					\item $\forall c\in C'\qquad (c,M'c)\in \rec (\hat\pi K)$
					\item $(v',w')\in \hat\pi K$
					\item $w'-v'\in C'$
				\end{itemize}
				Let $\gamma=w'-v'$. By Corollary \ref{cor:recProj}, there are $\gamma_1,\gamma_2,\gamma_3,\gamma_4\in \Rbb$ such that
				\centre{$(\gamma_1x+\gamma,\gamma_2x+M'\gamma)\in\rec (K)$}
				\lc{and}{$(\gamma_4x+M'\gamma,\gamma_3x+M'^2\gamma)\in\rec (K)$.}
				
				Using the function $s$ defined by Lemma~\ref{it:defS}, since $s(x)\in\Rbb x$, $s(-x)=-s(x)$
				(as argued in Point~(\ref{it:uxK})), and $M'\in\Rbb^*$, 
				we can assume without loss of generality that $\gamma_4=\gamma_1M'$, simply by adding a sufficient (possibly negative) multiple of $(x,s(x))$.
				Therefore
				\lcr{}{$(\gamma_1M'x+M'\gamma,\gamma_3x+M'^2\gamma)\in\rec (K)$}{$(\ast)$}	
				We select $C=\Vect (x)+\Rbbplus\gamma + \Rbbplus M'\gamma$. 
				\begin{itemize}
					\item If $\gamma=0$ we take any $M$ such that $Mx=s(x)$. 
					In this case $MC\subseteq \Rbb x=C$ and since $s(-x)=-s(x)\in\Ccal_u=C=\Rbb x$, we have for all $c\in C$, $(c,Mc)\in\rec (K)$.
					
					\item If $\gamma\neq 0$. We take $M$ such that 
					\centre{$Mx=s(x)\qqandqq M(\gamma_1x+\gamma)=\gamma_2x+M'\gamma$.}
					\lc{We have}{$M\gamma = M(\gamma_1x+\gamma) - \gamma_1Mx =  (\gamma_2x+M'\gamma) - \gamma_1s(x) \in C$}
					Note that, since $M'^2\gamma\in \Rbbplus\gamma$, we also have
					\centre{$MM'\gamma =M'\gamma_2x+M'^2\gamma - \gamma_1M's(x) \in C$}
					Therefore, $MC\subseteq C$. Moreover, we have
					\begin{calculs}
						& (x,Mx) &=& (x,s(x))\in\rec (K)\\
						& (\gamma,M\gamma) &=& (\gamma_1x+\gamma,\gamma_2x+M'\gamma)-\gamma_1(x,s(x))\in\rec (K)
					\end{calculs}
					If $M'\geq 0$ we then have for all $c\in C$, $(c,Mc)\in \rec (K)$. Otherwise, 
					dividing by $|M'|$ the $(\ast)$ statement,
					we have
					\centre{$\pa{-\gamma_1x-\gamma,\f{\gamma_3}{|M'|}x+|M'|\gamma}\in\rec (K)$}
					By conic combination,$\pa{0, \pa{\f{\gamma_3}{|M'|}+\gamma_2}x+(|M'|+M')\gamma}\in\rec (K)$.
					Hence $\pa{0, \pa{\f{\gamma_3}{|M'|}+\gamma_2}x}\in\rec (K)$.
					By assumption on $\Ccal_u$ this means that \center{$\gamma_3=-\gamma_2|M'|=\gamma_2M'$.}
					\begin{calculs}
						Thus & (M'\gamma,M'M\gamma) &=&
						(\gamma_1M'x+M'\gamma,\gamma_2M'x+M'^2\gamma) - M'\gamma_1(x,s(x))\\
						&&=&  (\gamma_1M'x+M'\gamma,\gamma_3x+M'^2\gamma) - 
						M'\gamma_1(x,s(x))
					\end{calculs}
					\lc{Hence, using $(\ast)$,}{$(M'\gamma,M'M\gamma)\in\rec (K)$}
					Thus, for all $c\in C$, $(c,Mc)\in \rec (K)$.
				\end{itemize}
				Therefore, in both cases, $M$ and $C$ satisfy Points \ref{it:MCC} and \ref{it:xMx}.

				Now, as $(v',w')\in \hat\pi K$, there are $a,b,\in E_u=\Vect (x)$, such that 
				\centre{$(a+v',b+w')\in K$.}
				\lc{Moreover}{$b+w'-a-v' = \underbrace{b-a}_{\in \Vect (x)} + \gamma\in C$}
				Hence, taking $v=a+v'$ and $w=b+w'$,
				$v,w,C$ satisfy Points \ref{it:vwz} and \ref{it:vwzC}.

				\item Finally, we consider the case where $\Ccal_u=\Rbb_+x$.
				Let $y\in E_u^\bot$ such that $\norm y{}=1$.
				Using what we saw at the beginning of Point (\ref{it:appelInduc}), if there is $(a,b)\in K$ such that
				$\inner{b-a,y}=0$, then there is $\mu\in\Rbb$ such that $b=a+\mu x$ which is impossible. Therefore,  
				for all $(a,b)\in K$, $|\inner{b-a,y}|>0$.
				Assume, for sake of contradiction, that there exist $a,b,c,d\in E$ such that
				\centre{$\inner{b-a,y}<0\qqandqq \inner{d-c,y}>0\qqandqq (a,b),(c,d)\in K$.}
				Let $\lambda=\f{\inner{d-c,y}}{\inner{d-c,y}-\inner{b-a,y}}$
				and $(e,f)=\lambda (a,b)+(1-\lambda)(c,d)\in K$.
				\lc{We then have}{$\inner{f-e,y} = \lambda\inner{b-a,y}+(1-\lambda)\inner{d-c,y} = 0$.}
				Therefore $f=e+\mu x$ for some $\mu\in\Rbb$. 
				
				Using the function $\delta$ defined by Lemma~\ref{it:defdelta} we have that 
				\[
				(e,f)+(1+|\mu|)(\delta(x),x+\delta(x))
				= (e+\delta(x),e+\delta(x))+(1+|\mu|+\mu)(0,x) \in K
				\]
				which contradicts the assumption of Point (\ref{it:appelInduc}).
				
				Therefore, either for all $(a,b)\in K$, $\inner{b-a,y}>0$ or for all $(a,b)\in K$, $\inner{b-a,y}<0$. Up to considering $-y$ instead of $y$, we assume that for all $(a,b)\in K$, $\inner{b-a,y}>0$. 
				
				Let $\mu = \inf\enstq{\inner{b-a,y}}{(a,b)\in K}$. We have $\mu>0$. Indeed,
				Note first that if $(w_1,w_2)\in\rec (K)$ then $\inner{w_2-w_1,y}\geq 0$ otherwise if would be easy to build a pair contradicting the assumption.
				Moreover, there exists $(v_1,v_2)\in K'$ such that 
				$\inner{v_2-v_1,y}$ achieves the minimum over $K'$: 
				\centre{$\inner{v_2-v_1,y}= \inf\enstq{\inner{b-a,y}}{(a,b)\in K}=\mu'$}
				$\mu'>0$ as $K'$ is compact. 
				For $(a,b)\in K$, writing it $(a,b)=(a',b')+(w_1,w_2)$ with $(a',b')\in K'$ and 
				$(w_1,w_2)\in\rec (K)$, we have
				\centre{$\mu\geq \inner{b-a,y}=\inner{b'-a',y}+\inner{w_2-w_1,y}\geq\inner{b'-a',y}\geq\mu'>0.$}

				\lc{In particular we have}{$\forall n\in\Nbb\qquad \inner{u_{n+1}-u_n,y}\geq \mu>0$,} 
				\lc{hence}{$\forall n\in\Nbb\qquad \inner{u_n,y}\geq\mu n+\inner{u_0,y}$}
				\lc{and}{$\inner{u_n,y}\tendsto n\pinf\pinf$.}
				Writing $u_n=\inner{u_n,x}x+\inner{u_n,y}y$, as $\Dcal_u=\{x\}$, we know that 
				\centre{$\inner{u_n,y}=\petito{n\to\pinf}{\inner{u_n,x}}$}
				\lc{and for $n$ large enough}{$\inner{u_n,x}>0$.}
				%Indeed, if this does not hold we would have $\Ccal_u\neq\Rbb_+x$. 
				Thus, up to considering a subsequence,  we can assume without loss of generality that for all $n\in\Nbb$, $\inner{u_n,x}>0$ and $\inner{u_n,y}>0$.
				Thus $\f{\inner{u_n,x}}{\inner{u_n,y}}$ is always defined and 
				\centre{$\f{\inner{u_n,x}}{\inner{u_n,y}}\tendsto n\pinf\pinf$.}
				Hence we can find an increasing function $\phi:\Nbb\to\Nbb$ such that
				\centre{$\forall n\in\Nbb\qquad \f{\inner{u_{\phi(n)+1},x}}{\inner{u_{\phi(n)+1},y}} > \f{\inner{u_{\phi(n)},x}}{\inner{u_{\phi(n)},y}}>1$.}
				\lc{Thus}{$\forall n\in\Nbb\qquad \f{\inner{u_{\phi(n)+1},x}}{\inner{u_{\phi(n)},x}} > \f{\inner{u_{\phi(n)+1},y}}{\inner{u_{\phi(n)},y}}$}
				Moreover, $\inner{u_{\phi(n)+1},y}\geq \mu + \inner{u_{\phi(n)},y} > \inner{u_{\phi(n)},y} >0$. Therefore
				\centre{$\forall n\in\Nbb\qquad \f{\inner{u_{\phi(n)+1},x}}{\inner{u_{\phi(n)},x}} > \f{\inner{u_{\phi(n)+1},y}}{\inner{u_{\phi(n)},y}}>1$}

				By applying Proposition \ref{prop:accExpRecK} to the sequence
				$\suite{\pa{\f{u_{n}}{\inner{u_{n+1},x}}, \f{u_{n+1}}{\inner{u_{n+1},x}}}}n\Nbb$
				we have that either $\suite{\f{\inner{u_{n+1},x}}{\inner{u_n,x}}}n\Nbb$ is bounded or
				$(0,x)\in\rec (K)$ which contradicts the assumption of the second main case of this proof.
				Therefore, up to considering a subsequence, we can assume that there exists $a\in\Rbb$ such that
				\centre{$\f{\inner{u_{n+1},x}}{\inner{u_n,x}}\tendsto n\pinf a\geq 1$.}
				
				By Proposition~\ref{prop:accExp}, we can obtain an accumulation expansion
				\centre{$\pa{u_{\phi(n)},u_{\phi(n)+1}} = \Sum{k=1}p\alpha_{k,n}(w_{k,1},w_{k,2}) + (w_{p+1,1},w_{p+1,2}) + \petito{n\to\pinf}1$}
				Note that since $\inner{u_n,y}=\petito{n\to\pinf}{\inner{u_n,x}}$ and $\inner{u_n,y}\tendsto n\pinf\pinf$,
				we necessarily have $p\geq 2$ and for $i\in\{1,2\}$, $(w_{i,1},w_{i,2}) \neq (0,0)$. Moreover, $w_{1,1},w_{1,2}\in\Ccal_u$.
				As 	$\Ccal_u= \Rbb_+ x$ and $\f{\inner{u_{n+1},x}}{\inner{u_n,x}}\tendsto n\pinf a$,	
				$w_{1,2}=aw_{1,1}$. Therefore,
				\centre{$w_{1,1},w_{1,2}\in\Rbbplusstar x$.}
				\lc{Similarly, since }{$\forall n\in\Nbb\qquad \f{\inner{u_{\phi(n)+1},x}}{\inner{u_{\phi(n)},x}} > \f{\inner{u_{\phi(n)+1},y}}{\inner{u_{\phi(n)},y}}$}
				the sequence $\suite{\f{\inner{u_{\phi(n)+1},y}}{\inner{u_{\phi(n)},y}}}n\Nbb$ is bounded. Therefore, $\inner{w_{2,1},y}\neq 0$. Hence, there are $b,c,d\in\Rbb$ and $\lambda\in \Rbbplusstar$ such that
				\centre{$(w_{2,1},w_{2,2})\in\lambda(dx+y,cx+by)$.}
				Indeed, if $\lambda$ was negative, we would have $\inner{u_{\phi(n)},y}<0$ for sufficiently large $n$, which is not possible as $\f{\inner{u_{\phi(n)+1},y}}{\inner{u_{\phi(n)},y}}>1$.
				Also, with this writing, we have 
				\centre{$\f{\inner{u_{\phi(n)+1},y}}{\inner{u_{\phi(n)},y}}\tendsto n\pinf b$}
				hence $b\leq a$.
				
				Now assume, for sake of contradiction, that for all $b'\geq 0$ and $c',d'\in\Rbb$ such that $ (d'x+y,c'x+b'y)\in\rec (K)$	we have $d'b'>c'$. 
				Let $a'\geq 0$ such that $(x,a'x)\in\rec (K)$ (note that $a'=a$ works as by 
				Proposition \ref{prop:accExpRecK} $(w_{1,2},w_{1,1})\in\rec (K)$).
				Then, for all $n\in\Nbb$, $\pa{(d'+n)x+y,(c'+na')x+b'y}\in\rec (K)$.
				\lc{Hence}{$(d'+n)b'>c'+na'$}
				This means that $b'\geq a'$ and as this holds for every $a',b'$, in particular, 	$b\geq a$ thus $b=a\geq 1$. As a consequence,
				\centre{$b=\min\enstq{b'}{\exists c',d'\ (d'x+y,c'x+b'y)\in\rec (K)}=\max\enstq{a'}{(x,a'x)\in \rec (K)}$.}
				Provided that $\inner{u_n,y}\tendsto n\pinf\pinf$, we have for large enough $n$, $(u_n,u_{n+1})\notin K'$. Writing
				\centre{$(u_n,u_{n+1})= (v_{n,1},v_{n,2}) + \lambda_n(d_nx+y,c_nx+b_n y)$}
				with $\lambda_n\geq 0$, $(v_{n,1},v_{n,2}) \in K'$ and $(d_nx+y,c_nx+b_n y)\in\rec (K)$, we have
				\centre{$\inner{u_{n+1},y} = b_n\lambda_n+\inner{v_{n,2},y}=b_n\pa{\inner{u_n,y}-\inner{v_{n,1},y}}+\inner{v_{n,2},y} 
					\underset{n\to\pinf}{\sim}b_n\inner{u_n,y}$}
				By minimality of $b$ and the fact that $\inner{u_n,y}\geq 0$ for large enough $n$, we get that for large enough $n$,
				\centre{$\inner{u_{n+1},y}\geq b\inner{u_n,y} + \petito{n\to\pinf}{\inner{u_n,y}}$.}
				\lc{Therefore}{$\inner{u_n,y} = \asOmega{n\to\pinf}{b^n}$,}
				\lc{hence}{$\f1{\inner{u_n,y}} = \grando{n\to\pinf}{\f1{b^n}}$.}
				Moreover, there exist sequences $(v_{1,n})_{n\in \Nbb}$,$(v_{2,n})_{n\in \Nbb}$, $\suiten[\alpha]$, $\suiten[\beta]$, $\suiten[\alpha']$, $\suiten[\beta']$ such that
				
				\centre{$(u_n,u_{n+1}) = (v_{1,n},v_{2,n}) + (\alpha_nx+\beta_n y, \alpha'_{n+1}x+\beta'_{n+1}y)$} 
				\lc{with}{$(v_{1,n},v_{2,n})\in K'\qqandqq (\alpha_nx+\beta_n y, \alpha'_{n+1}x+\beta'_{n+1}y)\in\rec (K)$.}
				\lc{Thus}{$\alpha_n,\alpha'_n = \inner{u_n,x}+\grando{n\to\pinf}1$}
				\lc{and}{$\beta_n,\beta'_n = \inner{u_n,y}+\grando{n\to\pinf}1$.}
				By the earlier assumption of this contradiction proof applied on
				$\f{\beta'_{n+1}}{\beta_{n}}, \f{\alpha'_{n+1}}{\beta_{n}}$ and $\f{\alpha_{n}}{\beta_{n}}$ 
				we have $\alpha_n\beta'_{n+1}>\alpha'_{n+1}\beta_n$
				\lc{and}
				{$\f{\alpha_n}{\inner{u_n,y}}\f{\beta'_{n+1}}{\inner{u_{n+1},y}} 
					> \f{\alpha'_{n+1}}{\inner{u_{n+1},y}}\f{\beta_n}{\inner{u_n,y}}$.}
				\begin{calculs}
					Hence & \f{\f{\inner{u_n,x}}{\inner{u_n,y}} + \grando{n\to\pinf}{\f1{\inner{u_n,y}}}}{ 1 + \grando{n\to\pinf}{\f1{\inner{u_n,y}}}} &>&
					\f{\f{\inner{u_{n+1},x}}{\inner{u_{n+1},y}} + \grando{n\to\pinf}{\f1{\inner{u_{n+1},y}}}}{ 1 + \grando{n\to\pinf}{\f1{\inner{u_{n+1},y}}}}\\[1cm]
					thus & \f{\inner{u_n,x}}{\inner{u_n,y}} + \grando{n\to\pinf}{\f1{\inner{u_n,y}}} &>&
					\f{\inner{u_{n+1},x}}{\inner{u_{n+1},y}} + \grando{n\to\pinf}{\f1{\inner{u_{n+1},y}}}\\ [1cm]
					finally & \f{\inner{u_{n+1},x}}{\inner{u_{n+1},y}} &<& \f{\inner{u_n,x}}{\inner{u_n,y}} +\grando{n\to\pinf}{\f1{b^n}}
				\end{calculs}
				
				As $\f{\inner{u_n,x}}{\inner{u_n,y}}\tendsto n\pinf\pinf$, we cannot have $b>1$. 	
				Thus $a=b=1$.			
				%If $b>1$, this implies that the sequence 	$\suite{\f{\inner{u_n,x}}{\inner{u_n,y}}}n\Nbb$ is bounded, which
				%
				%Since $\f{\inner{u_n,x}}{\inner{u_n,y}}>0$ for all $n$, we cannot have $b>1$. Indeed, this would imply that the sequence
				%					$\suite{\f{\inner{u_n,x}}{\inner{u_n,y}}}n\Nbb$ is bounded, which is a contradiction. Therefore,
				%					\centre{$a=b=1$}

				Recall that $\pa{u_{\phi(n)},u_{\phi(n)+1}} = \Sum{k=1}p\alpha_{k,n}(w_{k,1},w_{k,2}) + (w_{p+1,1},w_{p+1,2}) + \petito{n\to\pinf}1$, $aw_{1,1}=w_{1,2}$ and $(w_{2,1},w_{2,2})\in\lambda(dx+y,cx+by)$ with $\lambda\in Rbb_+$.
				As $a=b=1$, there exists $w\in \Rbb$ such that $\inner{w_{1,1},x}=\inner{w_{1,2},x}=w$ 
				and we get
				\centre{$\begin{accolade}
						\inner{u_{\phi(n)},x} = \alpha_{1,n}w + \lambda d\alpha_{2,n} + \petito{n\to\pinf}{\alpha_{2,n}}\\
						\inner{u_{\phi(n)+1},x} = \alpha_{1,n}w + \lambda c\alpha_{2,n} + \petito{n\to\pinf}{\alpha_{2,n}}
					\end{accolade}$}
				\lc{Thus}{$\inner{u_{\phi(n)+1},x} - \inner{u_{\phi(n)},x} = \lambda(c-d)\alpha_{2,n} + \petito{n\to\pinf}{\alpha_{2,n}} \tendsto n\pinf\minf$}
				as $c<bd = d$, which is a contradiction with the fact that  for all $n\in\Nbb$, $\f{\inner{u_{\phi(n)+1},x} }{\inner{u_{\phi(n)},x}}>1$.
				
				Thus, there are $a,b\geq 0$ and $c,d\in\Rbb $ such that
				\centre{$(x,ax)\in\rec (K)\qqandqq (dx+y,cx+by)\in\rec (K)\qqandqq c\geq db$}
				
				We consider $C=\Rbbplus x + \Rbbplus(dx+y)$. We also take the matrix $M$ such that
				$Mx =ax$ and $M(dx+y)=cx+by$.
				Since $cx+by = b(dx+y)+(c-db)x$, and $c\geq db$, we indeed have $MC\subseteq C$.
				Hence, $M$ and $C$ satisfy Points \ref{it:MCC} and \ref{it:xMx}.
				Assume now that for all $n\in\Nbb$, 
				\centre{$\inner{u_{n+1}-u_n,x}<d\inner{u_{n+1}-u_n,y}$}
				\lc{Then for all $n\in\Nbb$}{$\inner{u_n-u_0,x}<d\inner{u_n-u_0,y}$}
				And this is a contradiction with $\f{\inner{u_n,x}}{\inner{u_n,y}}\tendsto n\pinf\pinf$ and $\inner{u_n,y}>0$.
				Thus let $n\in\Nbb$ such that $\inner{u_{n+1}-u_n,x}\geq d\inner{u_{n+1}-u_n,y}$.
				Let $v=u_n$ and $w=u_{n+1}$.
				\begin{calculs}
					& w-v &=& u_{n+1}-u_n = \inner{u_{n+1}-u_n,x}x + \inner{u_{n+1}-u_n,y}y\\
					&&=& \pa{\inner{u_{n+1}-u_n,x}-d\inner{u_{n+1}-u_n,y}}x + \inner{u_{n+1}-u_n,y}(dx+y)\in C
				\end{calculs}
				Hence, Points \ref{it:vwz} and \ref{it:vwzC} are satisfied by $C,v,w$.
			\end{romanenumerate}
		\end{alphaenumerate}
	\end{itemize}
\end{proof}

%% file: main.bbl
\begin{thebibliography}{1}

\bibitem{Ben-AmramDG19}
Amir~M. Ben{-}Amram, Jes{\'{u}}s~J. Dom{\'{e}}nech, and Samir Genaim.
\newblock Multiphase-linear ranking functions and their relation to recurrent
  sets.
\newblock In {\em Static Analysis - 26th International Symposium, {SAS} 2019,
  Proceedings}, volume 11822 of {\em Lecture Notes in Computer Science}, pages
  459--480. Springer, 2019.

\bibitem{BG14}
Amir~M. Ben-Amram and Samir Genaim.
\newblock Ranking functions for linear-constraint loops.
\newblock {\em Journal of the {ACM}}, 61(4):1--55, 2014.
\newblock \href {https://doi.org/10.1145/2629488} {\path{doi:10.1145/2629488}}.

\bibitem{BGM12}
Amir~M. Ben-Amram, Samir Genaim, and Abu~Naser Masud.
\newblock On the termination of integer loops.
\newblock {\em {ACM} Transactions on Programming Languages and Systems},
  34(4):1--24, 2012.
\newblock \href {https://doi.org/10.1145/2400676.2400679}
  {\path{doi:10.1145/2400676.2400679}}.

\bibitem{bozga}
Marius Bozga, Radu Iosif, and Filip Konecn{\'{y}}.
\newblock Deciding conditional termination.
\newblock {\em Log. Methods Comput. Sci.}, 10(3), 2014.
\newblock \href {https://doi.org/10.2168/LMCS-10(3:8)2014}
  {\path{doi:10.2168/LMCS-10(3:8)2014}}.

\bibitem{Braverman06}
Mark Braverman.
\newblock Termination of integer linear programs.
\newblock In {\em Computer Aided Verification 2006}, volume 4144 of {\em LNCS},
  pages 372--385. Springer Berlin Heidelberg, 2006.
\newblock \href {https://doi.org/10.1007/11817963_34}
  {\path{doi:10.1007/11817963_34}}.

\bibitem{HOW19}
Mehran Hosseini, Jo\"{e}l Ouaknine, and James Worrell.
\newblock Termination of linear loops over the integers.
\newblock In {\em ICALP 2019}, volume 132 of {\em LIPIcs}, pages 118:1--118:13.
  Schloss Dagstuhl - Leibniz-Zentrum fuer Informatik GmbH, Wadern/Saarbruecken,
  Germany, 2019.
\newblock \href {https://doi.org/10.4230/LIPICS.ICALP.2019.118}
  {\path{doi:10.4230/LIPICS.ICALP.2019.118}}.

\bibitem{LeikeH18}
Jan Leike and Matthias Heizmann.
\newblock Geometric nontermination arguments.
\newblock In {\em Tools and Algorithms for the Construction and Analysis of
  Systems - 24th International Conference, {TACAS} 2018}, volume 10806 of {\em
  Lecture Notes in Computer Science}, pages 266--283. Springer, 2018.

\bibitem{LindenstraussS97}
Naomi Lindenstrauss and Yehoshua Sagiv.
\newblock Automatic termination analysis of logic programs.
\newblock In Lee Naish, editor, {\em Logic Programming, Proceedings of the
  Fourteenth International Conference on Logic Programming, 1997}, pages
  63--77. {MIT} Press, 1997.

\bibitem{Tiwari04}
Ashish Tiwari.
\newblock Termination of linear programs.
\newblock In {\em Computer Aided Verification 2004}, volume 3114 of {\em LNCS},
  pages 70--82. Springer Berlin Heidelberg, 2004.
\newblock \href {https://doi.org/10.1007/978-3-540-27813-9_6}
  {\path{doi:10.1007/978-3-540-27813-9_6}}.

\end{thebibliography}
